\documentclass[12pt,a4paper]{article}
\usepackage{url}
\usepackage[utf8]{inputenc}
\usepackage{tikz}
\usepackage{cancel}
\usepackage{color}
\usepackage[normalem]{ulem}
\usepackage[english]{babel}
\usepackage{color}
\usepackage[english]{babel}
\usepackage{color}
\usepackage{graphicx}
\usepackage{amsmath,amssymb,amsthm}
\usepackage[T1]{fontenc}
\usepackage[utf8]{inputenc}
\usepackage{authblk}
\usepackage[english]{babel}
\usepackage{amsmath,amssymb,amsthm}
\usepackage{color}
\usepackage{graphics}
\usepackage{amsmath,amssymb,amsthm}
\usepackage{lipsum}
\newcommand*\rectangled[1]{%
      \tikz[baseline=(R.base)]\node[draw,rectangle,inner sep=0.9pt](R) {#1};\!
}
\newcommand*\rund[1]{%
      \tikz[baseline=(R.base)]\node[draw,circle,inner sep=0.5pt](R) {#1};\!
}

\newtheorem{theorem}{Theorem}
\newtheorem{proposition}[theorem]{Proposition}
\newtheorem{corollary}[theorem]{Corollary}
\newtheorem{lemma}[theorem]{Lemma}
\theoremstyle{definition}
\newtheorem{definition}[theorem]{Definition}
\newtheorem{remark}[theorem]{Remark}
\newtheorem{example}{Example}

\renewcommand{\tilde}[1]{\widetilde{#1}}

\newcommand\xqed[1]{%
  \leavevmode\unskip\penalty9999 \hbox{}\nobreak\hfill
  \quad\hbox{#1}}
\newcommand\demo{\xqed{$\Diamond$}}

\title{Improved constructions of nested code pairs}
\author[1]{Carlos Galindo\thanks{galindo@mat.uji.es}}
\author[2]{Olav Geil\thanks{olav@math.aau.dk}}
\author[1,2]{Fernando Hernando\thanks{carrillf@mat.uji.es}}
\author[2]{Diego Ruano\thanks{diego@math.aau.dk}}
\affil[1]{Instituto Universitario de Matemáticas y Aplicaciones de Castellón, and
Departamento de Matemáticas, Jaume I University, Spain}
\affil[2]{Department of Mathematical Sciences, Aalborg University, Denmark}

\begin{document}
\maketitle

\begin{abstract}
Two new constructions of linear code pairs $C_2 \subset C_1$ are given for
which the codimension and the relative minimum distances $M_1(C_1,C_2)$,
$M_1(C_2^\perp,C_1^\perp)$ are good. By this we
mean that for any two out of the three parameters
the third parameter of the
constructed code pair is large. Such pairs of nested codes are indispensable for the
determination of good linear ramp secret sharing schemes \cite{kurihara}. They can
also be used to ensure reliable communication over asymmetric quantum
channels \cite{steane1996simple}. The new constructions result from
carefully applying the
Feng-Rao bounds \cite{FR1,agismm} to a family of codes defined from multivariate
polynomials and Cartesian product point sets. \\

\noindent {\bf{Keywords:}} asymmetric quantum code, CSS construction,
Feng-Rao bound, nested codes, ramp secret sharing, relative
generalized Hamming weight, relative minimum distance, wiretap channel of type II.
\end{abstract}

\section{Introduction}\label{secintro}
In this paper we consider pairs of linear codes $C_2 \subset C_1 \subseteq {\mathbb{F}}_q^n$
where ${\mathbb{F}}_q$ is the finite field with $q$ elements. We are
interested in the codimension $\ell=\dim C_1-\dim C_2$ and the
relative minimum distances
$$M_1(C_1,C_2)=\min \{ w_H(\vec{c}) \mid \vec{c} \in C_1 \backslash
C_2\},$$
$$M_1(C_2^\perp,C_1^\perp)=\min \{ w_H(\vec{c}) \mid \vec{c} \in C_2^\perp \backslash
C_1^\perp\}.$$
Here
$w_H(\vec{c})$ means the Hamming weight of $\vec{c}$. For any two out of
three parameters we aim to construct code pairs such that the
two 
parameters are equal to some prescribed values, whereas the last parameter
is as large as possible. Our motivation for studying the above problem is
applications in ramp secret sharing, communication over wiretap
channels of type II,  and asymmetric quantum coding. \\

We first explain the application in secret sharing. The application to wiretap channels of type II is analogue. A secret sharing scheme is a cryptographic method to encode a secret
into a set of shares, later to be distributed among participants, so that only
specified subsets of the participants can reconstruct the secret. The first
secret sharing scheme, proposed by Shamir~\cite{shamir1979share}, was a
perfect scheme, meaning that a set of participants unable to reconstruct
the secret has no information on the secret. Later
non-perfect secret sharing schemes were proposed~\cite{Blakley,Yamamoto} in which
there are sets of participants that have some 
information about the secret, but cannot fully reconstruct it. In this paper
we use the term ramp
secret sharing schemes for the general class of perfect or non-perfect schemes. 
Secret sharing
has been applied, for example, to store confidential information at
multiple locations that are placed  geographically apart. When we use  secret sharing
schemes in such a scenario, the likelihoods of both data loss and data
theft are decreased. As far as we know, in many applications both
perfect and non-perfect ramp secret sharing schemes are useful. In
the perfect scheme the size of a share must be greater than or equal to that of the
secret~\cite{capocelli1993size}. In contrast ramp secret sharing schemes allow
shares to be smaller than the secret which for instance is useful for storing bulk
data~\cite{csirmaz2009ramp}.

A linear ramp secret sharing scheme
can be described
as a coset
construction $C_1/C_2$ where $C_2 \subset C_1$ are linear codes
\cite{Chen}. More precisely, let $\dim C_2=k_2$, $\dim C_1=k_1$ and
$\ell=k_1-k_2$. Given a basis $\{\vec{b}_1, \ldots ,
\vec{b}_{k_2} \}$ for $C_2$
as a vector space over ${\mathbb{F}}_q$ and a corresponding basis $\{
\vec{b}_1, \ldots , \vec{b}_{k_2}, \vec{b}_{k_2+1}, \ldots ,
\vec{b}_{k_1=k_2+\ell}\}$ for $C_1$ the encoding of a secret $(s_1,\ldots
, s_\ell) \in {\mathbb{F}}_q^{\ell}$ is done by choosing $a_1,
\ldots , a_{k_2} \in {\mathbb{F}}_q$ randomly from a uniform distribution and then constructing
the codeword $\vec{c}=a_1\vec{b}_1+\cdots +a_{k_2}\vec{b}_{k_2}+s_1 \vec{b}_{k_2 +1}+\cdots +s_\ell \vec{b}_{ k_1}$. The
shares are the entries of $\vec{c}$. 

\begin{definition}\label{def:privreco}
Given a ramp secret sharing scheme $C_1/C_2$ with $\ell = \dim C_1-\dim C_2$
we say that it has
$(t_1,\ldots,t_\ell)$-privacy and $(r_1, \ldots ,
r_\ell)$-reconstruction if the positive integers $t_1,\ldots , t_{\ell}$ are
chosen as large as
possible and the positive integers $r_1, \ldots
,r_{\ell}$ are chosen as small as possible such that
\begin{itemize}
\item for $1 \leq v \leq \ell$, an adversary cannot obtain $v \log_2(q)$ bits of
  information  about $\vec{s}$ with any
  $t_v$ shares, 
\item for $1 \leq v \leq \ell$, it is possible to recover $v \log_2(q)$ bits of information about $\vec{s}$ with any
  collection of $r_v$ shares. 
\end{itemize}
\end{definition}

We shall refer to the numbers $r_1, \ldots , r_\ell$ as
reconstruction numbers and similarly call the numbers $t_1, \ldots ,
t_\ell$ privacy numbers. These parameters are motivated by the fact that the amount of
information which an adversary can obtain is always an integer times
$log_2(q)$ bits and similar for the reconstruction. Of particular interest are the first
privacy number $t=t_1$ and the last reconstruction number $r=r_\ell$, as
$t$ is the maximal number such that no set of this size leaks any
information about the secret, and $r$ equals the smallest 
number such that any set of this size can recover the entire secret. It was demonstrated
in~\cite{Bains,subramanian2009mds,kurihara,geil2014relative} 
that the above
numbers can be uniquely determined from the relative generalized Hamming
weights, that we shall define now. For $v=1, \ldots , \ell$
\begin{eqnarray}
M_v(C_1,C_2)&=&\min \{\# {\mbox{\textnormal{Supp }}} U \mid U {\mbox{ is a subspace
      of }} C_1\nonumber \\
&& {\mbox{ {\hspace{3cm}} of dimension }} v, U \cap C_2=\{ \vec{0} \}
                \}\nonumber 
\end{eqnarray}
(and similar for the dual codes). 
Here, ${\mbox{\textnormal{Supp} }} U$ is the set of  entries where some codeword in
$U$ is non-zero and $\#$ is the cardinality. In our paper we shall adopt the tradition of sometimes
referring to   
relative generalized Hamming weights
$M_1(C_1,C_2)$ and $M_1(C_2^\perp, C_1^\perp)$ simply as relative
minimum 
distances. Note that the relative minimum distance $M_1(C_1,C_2)$ can be
lower bounded by the minimum distance of $C_1$. Similarly, the
relative minimum distance $M_1(C_2^\perp,C_1^\perp)$ is greater than or equal
to the minimum distance of $C_2^\perp$. The following theorem corresponds to~\cite[Theorem 3]{geil2014relative}.
\begin{theorem}
Let $C_2\subset C_1 \subseteq {\mathbb{F}}_q^n$ with $\ell =\dim
C_1-\dim C_2$ and consider the corresponding ramp secret sharing
scheme $C_1/C_2$. Then the reconstruction
numbers and privacy numbers satisfy
\begin{eqnarray}
r_v&=&n-M_{\ell - v+1}(C_1,C_2)+1, \nonumber \\
t_v&=&M_v(C_2^\perp, C_1^\perp)-1, \nonumber 
\end{eqnarray}
for $v=1, \ldots , \ell$.
\end{theorem}
Hence, if for instance we want to construct a ramp secret sharing scheme over ${\mathbb{F}}_q$ with $n$
participants, 
secrets of length
$\ell$, first privacy number equal to some $t$,  and last reconstruction
number $r$ as small as possible, what we need is exactly a pair
of nested codes $C_2\subset C_1 \subseteq {\mathbb{F}}_q^n$ of
codimension $\ell$ with $M_1(C_2^\perp,C_1^\perp)=t+1$, and 
$M_1(C_1,C_2)=n-r+1$ as large as possible.\\

Finally we explain in brief the use of nested codes in connection with asymmetric
quantum error-correcting codes, introduced
in~\cite{steane1996simple}. The study of good quantum 
codes is by now a
well-established research
area. Some classical and recent references are~\cite{q8,q6,q7,q3,q5,q2,q14,kkk}. Recently, such theory has been
extended to asymmetric quantum error-correcting codes which are
useful in a model where the probabilities of qubit-flip and phase-shift errors
are
different~\cite{aq7,aq6,aq5,aq1,aq2,rev3_4,ezerman2013css,aq3,aqlaguardia}. This
generalization is
motivated by the argument that dephasing will happen more frequently
than relaxation~\cite{aq7}. A linear $q$-ary asymmetric quantum
error-correcting code 
$C$ is a $q^k$ dimensional subspace of the Hilbert space
${\mathbb{C}}^{q^n}$ whose error basis is defined by unitary operators
usually denoted by $X$ and $Z$. It is customary to write the
parameters of $C$ as $[[n,k,d_z/ d_x]]_q$ which means that $C$
corrects all phase-shift errors up to $\lfloor
\frac{d_z-1}{2}\rfloor$, and all qudit-flip errors up to $\lfloor \frac{d_x-1}{2} \rfloor$.

In the present paper we concentrate on
the Calderbank-Shor-Steane (CSS) construction of asymmetric quantum
codes from a pair of nested linear codes $C_2 \subset C_1
\subseteq {\mathbb{F}}_q^n$. We leave it for the reader to
inspect~\cite{q7,kkk} for the actual construction. Here, we only give the
following important result on the parameters of the resulting
asymmetric quantum
code (see \cite[Lemma 3.1]{aq6}). 

\begin{theorem}\label{thecss}
Consider linear codes $C_2 \subset C_1 \subseteq
{\mathbb{F}}_q^n$. Then the asymmetric quantum code defined
using the CSS construction has parameters $$[[n,\ell=\dim C_1-\dim C_2,
d_z / d_x]]_q$$ where $d_z=M_1(C_1,C_2)$ and $d_x=M_1(C_2^\perp,C_1^\perp)$.
\end{theorem}

Recall, that a stabilizer (symmetric) quantum code is the common eigenspace of a
commutative subgroup of the error group associated to the error basis
(see~\cite{q7,kkk} for details). The quantum codes in Theorem~\ref{thecss}
can be considered as stabilizer asymmetric quantum codes~\cite[Lemma
3.1]{aq6}. Studying asymmetric quantum codes rather than only
symmetric codes is an important problem. For instance, already 
in~\cite{aq7} it was identified that large ratios $d_z/d_x$ are relevant
 -- phase-flip
errors occurring tens, hundreds, or even thousands times more likely than
bit-flips. From Theorem~\ref{thecss} it is clear that $d_z\geq
d(C_1)$ and $d_x\geq d(C_2^\perp)$ where the expressions on the right
sides are the minimum distance of the classical code. There is a clear
physical significance of cases where strict inequality holds in at
least one of these expressions. Such (asymmetric) quantum codes are
called impure (or degenerate)~\cite{rev3_3,recentpreprint} and it is known that the impureness
can be employed to obtain improved decoding. 

As a measure for goodness of asymmetric quantum codes
we shall use the Gilbert-Varshamov bound from~\cite[Theorem 4]{MGV} which
we now recall:
\begin{theorem}\label{theMGV}
If
$$\frac{1-q^{-2\ell}}{1-q^{-2n}} \cdot \frac{1}{q^{n-\ell}} \sum
_{i=1}^{d_x-1} {n \choose i }(q-1)^i \cdot \sum
_{i=1}^{d_z-1} {n \choose i }(q-1)^i < 1$$
then there exists an $[[n,\ell,d_z/d_x]]_q$ asymmetric quantum code. 
\end{theorem}
The literature only reports few families of long asymmetric quantum
codes, an exception being La Guardia's construction II in~\cite[Theorem
7.1]{aqlaguardia} which we shall compare our codes with. In another
direction Ezerman et al's works in~\cite{ezerman2013css,rev3_5} explain
how to derive good asymmetric quantum codes with $d_x$ being very
small and $\ell$ being moderate. As our
constructions do not seem to generally compare well with these codes
in the case of $d_x \in \{2,3\}$ we omit such cases in our tables.

Having discussed both ramp secret sharing schemes and asymmetric
quantum codes we include a remark relating them to each other.

\begin{remark}\label{remsammenhaeng}
There exists an asymmetric quantum code based on the CSS construction
with parameters $[[n,\ell,d_z/d_x]]_q$ if and only there exists a
linear ramp secret sharing scheme over ${\mathbb{F}}_q$ with secrets in
${\mathbb{F}}_q^\ell$, with $r=r_\ell=n-d_z+1$ and $t=t_1=d_x-1$.
\end{remark}
As the tradition for reporting numerical data on parameters of
(asymmetric) quantum codes seems stronger than the tradition for
reporting corresponding parameters of ramp secret sharing schemes,
throughout this paper we shall often report our findings in the first
setting. In such cases we leave it for the reader to apply
Remark~\ref{remsammenhaeng} to make the translation to secret
sharing. On the other hand, higher relative weights give information
on ramp secret sharing schemes, whereas  no implication for
asymmetric quantum codes seems to be
known. Hence, when treating them we shall do it in the setting of
secret sharing.\\

In this paper we present two new families of long nested codes $C_2 \subset C_1$ with
$M_1(C_1,C_2)+M_1(C_2^\perp, C_1^\perp)$ high. Such codes give rise to
ramp secret sharing schemes with $r-t$ close to $\ell$, and give rise
to asymmetric quantum codes with $d_z+d_x$ close to $n-\ell+2$. The code pairs are defined by evaluating
multivariate polynomials at the points of Cartesian products of
subsets of finite fields, and the above mentioned two families
are found by carefully applying Feng-Rao theory~\cite{geil2014relative}. Our first
family is made by combining, for the first time in the literature, the Feng-Rao improved code
constructions for dual and primary codes. This leads to
good pairs of codes, however, it only works for relatively high codimension
$\ell$. The asymmetric quantum codes related to this first family of
nested codes compare very favorably
with known asymmetric quantum codes of similar length as well as the Gilbert-Varshamov
bound. Moreover, the construction is very flexible and we can choose
$d_z/d_x$ very large, which as already mentioned can be desirable. Our second
family is a completely new construction which produces very good
parameters in the case of relatively small codimension $\ell$. Again
the corresponding asymmetric quantum codes compare favorably with the
known codes of similar length and similarly with the Gilbert-Varshamov
bound. Even more, we demonstrate a strong advantage
of using our estimates on the relative minimum distances
$M_1(C_1,C_2)$ and $M_1(C_2^\perp,C_1^\perp)$; rather than
just using information on the
minimum distances $d(C_1)$ and $d(C_2^\perp)$. Actually, using only information on the minimum
distances $d(C_1)$ and $d(C_2^\perp)$  -- which is often done in the literature -- it seems in many cases impossible to establish the code
parameters for asymmetric quantum codes, which we are able to
obtain. These reflections are closely related to the fact that the corresponding
asymmetric quantum codes
of relatively small dimension are almost always impure, which
as already mentioned is desirable. Again the construction is quite
flexible in the sense that we can choose the ratio $d_z/d_x$ very
high if requested. For both families of codes 
we provide generator matrices, and we also describe a method for estimating
higher relative weights $M_v(C_1, C_2)$, $M_v(C_2^\perp,C_1^\perp)$, $v=2,
\ldots , \ell$, which sometimes leads to closed formula
expressions. Recall, that such parameters express the information
leakage and message recovery in connection with ramp secret
sharing. As a result it is shown that for our second family of nested
codes, the security of the related secret sharing schemes is much
better than expected from studying only $t=t_1$. Furthermore, for certain choices of Cartesian
product point sets also parity check matrices
can easily be obtained, namely for the particular cases where the
considered codes satisfy the conditions for being so-called
$J$-affine variety codes~\cite{galindo2015quantum}. Finally, all
considered codes in this paper can be decoded up to half their
designed minimum distance by applying known decoding algorithms. The
dual codes can furthermore be decoded up to half the designed relative
minimum distance.  These observations lead to decoding algorithms for
the corresponding asymmetric quantum codes.

The paper is organized as follows. In Section~\ref{sec4} we start by recalling the Feng-Rao bounds for
primary and dual linear codes and we apply them to the general class
of  
codes derived by evaluating multivariate polynomials at Cartesian
product point sets. In this section we provide all needed background on Feng-Rao
theory -- for basic results on multivariate polynomials and related
concepts we refer the reader to~\cite{clo}. The section concludes with
a discussion on how to decode the related asymmetric quantum codes. Then, in 
Section~\ref{seclarge}, we explain how to employ the
Feng-Rao improved code constructions for primary and dual codes
simultaneously to obtain good families of nested codes with relatively
high codimension. We then study the corresponding ramp secret sharing
schemes and asymmetric quantum codes. In Section~\ref{secsmall} we present
the new good construction
of nested codes with relatively small
codimension and we study the corresponding ramp secret sharing schemes
and asymmetric quantum codes. Section~\ref{secon} 
gives concluding remarks on the connection to $J$-affine
variety codes. The paper contains a number of examples, the end of which we
indicate by $\Diamond$s.

\section{Codes defined from Cartesian product point sets}\label{sec4}
In this paper we consider families of code pairs $C_2
\subset C_1$  defined by evaluating multivariate polynomials at the
points of 
Cartesian products of subsets of finite fields. For the applications described in the previous section,
we are interested in the parameters $M_v(C_1,C_2)$ (the primary case) as well as
$M_v(C_2^\perp, C_1^\perp)$ (the dual case) -- with a special focus
on the relative minimum distances $M_1(C_1,C_2)$ and $M_1(C_2^\perp,C_1^\perp)$. To handle the primary
case only it would be natural to use the language of Gr\"{o}bner
basis theory and to apply the so-called
footprint bound~\cite{sakata1990extension,onorin,geilhoeholdt}. However, in this language
it is more difficult to treat the dual case and we
therefore give a coherent description of both cases using the Feng-Rao
bounds for general linear codes instead. The Feng-Rao bounds come in two
versions, namely one for
primary codes~\cite{AG,MR2831617,agismm,geil2014relative} and another for
dual
codes~\cite{FR24,FR1,FR2,early,handbook,MM,geil2014relative}.\\

 Our
exposition follows the presentation in~\cite[Section IV]{geil2014relative}
and is illustrated with a continued example. This continued example,
at the end of the present 
section, leads to the introduction of a general class of code pairs for which we
have a simple description of generator matrices,  where we know
the codimension, and where we can easily estimate the relative minimum distances and
also the higher relative weights (Theorem~\ref{theour}). It is from
this class of code pairs we, in the following sections, show how to
choose optimal pairs when the codimension is relatively large
(Section~\ref{seclarge}),  and when it is relatively small (Section~\ref{secsmall}).\\ 

\noindent Consider a fixed basis ${\mathcal{B}}=\{\vec{b}_1, \ldots ,
\vec{b}_n\}$ for ${\mathbb{F}}_q^n$ as a vector space over
${\mathbb{F}}_q$ and let ${\mathcal{I}}=\{1, \ldots , n\}$. 
\begin{definition}\label{defro}
Define $\bar{\rho} : {\mathbb{F}}_q^n \rightarrow \mathcal{I}
\cup \{ 0 \}$ to be the function given as follows. For non-zero $\vec{c}$ we have
$\bar{\rho}(\vec{c})=i$ where $i$ is the unique integer such that
$$\vec{c} \in {\mbox{Span}} \{ \vec{b}_1, \ldots , \vec{b}_i\}
\backslash {\mbox{Span}} \{ \vec{b}_1, \ldots ,
\vec{b}_{i-1}\}.$$ Here we use the convention that ${\mbox{Span}} \, 
\emptyset = \{ \vec{0}\}$. Finally, let $\bar{\rho}(\vec{0})=0$.
\end{definition}
Throughout the paper by $\prec_{\deg}$ we shall always mean the degree
lexicographic ordering given by the rule that for two different
monomials we have $X_1^{i_1} \cdots X_m^{i_m} \prec_{\deg} X_1^{j_1}
\cdots X_m^{j_m}$ if one of the following conditions holds:
\begin{enumerate}
\item $i_1+\cdots +i_m < j_1 +\cdots +j_m$
\item$ i_1+\cdots + i_m = j_1 +\cdots +j_m$, but the
               rightmost non-zero entry of\\
  $(j_1-i_1, \ldots , j_m-i_m)$  is
  positive. 
\end{enumerate}

In case of two variables $X$ and $Y$, we shall always think of $X$ as
$X_1$ and $Y$ as $X_2$. In the paper we shall also need other monomial orderings
$\prec$, however, the degree lexicographic ordering will play a particular
important role.

\begin{example}\label{ex1}
Consider the ideal $I=\langle X^6-1,Y^6-1\rangle \subset
{\mathbb{F}}_7[X,Y]$ and the residue class ring $R={\mathbb{F}}_7[X,Y]/I$. The corresponding
variety consists of all pairs of non-zero elements of
${\mathbb{F}}_7$, hence we may write ${\mathcal{V}}_{\mathbb{F}_q}(I)=\{P_1,  \ldots , P_{36}\}$. Let
${\mbox{\textnormal{ev}}} : R  \rightarrow   {\mathbb{F}}_7^{36}$ be the vector
space homomorphism given by ${\mbox{\textnormal{ev}}}(F+I)=(F(P_1), 
\ldots , F(P_{36}))$. The set ${\mathcal{B}}=\{ {\mbox{\textnormal{ev}}}
(X^iY^j+I) \mid 0 \leq i,j< 6\}$ then constitutes a basis for
${\mathbb{F}}_7^{36}$ as a vector space over ${\mathbb{F}}_7$. To see
this, we first observe that 
\begin{eqnarray}
{\mbox{\textnormal{ev}}}(F(X,Y)+I)&=& {\mbox{\textnormal{ev}}} \big( F(X,Y) -A(X,Y)(X^6-1)-B(X,Y)(Y^6-1)+I\big) \nonumber 
\end{eqnarray}
 for any $A(X,Y), B(X,Y)\in {\mathbb{F}}_7[X,Y]$, which implies that we
may, without loss of generality, assume that $\deg_X F,
\deg_Y F < 6$. Using Lagrange interpolation it holds that ${\mbox{\textnormal{ev}}}$
is surjective, and as $\#{\mathcal{B}}$ equals the dimension of the
image 
${\mathbb{F}}_7^{36}$ ${\mathcal{B}}$ is indeed a basis -- and 
consequently ${\mbox{\textnormal{ev}}}$ is an isomorphism. We next enumerate ${\mathcal{B}}$
according to the degree lexicographic ordering $\prec_{\deg}$.
The enumeration is illustrated in Figure~\ref{fig0}. 
\begin{figure}[h]
$$
\begin{array}{cccccc}
Y^5&XY^5&X^2Y^5&X^3Y^5&X^4Y^5&X^5Y^5\\
Y^4&XY^4&X^2Y^4&X^3Y^4&X^4Y^4&X^5Y^4\\
Y^3&XY^3&X^2Y^3&X^3Y^3&X^4Y^3&X^5Y^3\\
Y^2&XY^2&X^2Y^2&X^3Y^2&X^4Y^2&X^5Y^2\\
Y&XY&X^2Y&X^3Y&X^4Y&X^5Y\\
1&X&X^2&X^3&X^4&X^5
\end{array}
{\mbox{ \ \ \ \ }}
\begin{array}{cccccc}
\vec{b}_{21}&\vec{b}_{26}&\vec{b}_{30}&\vec{b}_{33}&\vec{b}_{35}&\vec{b}_{36}\\
\vec{b}_{15}&\vec{b}_{20}&\vec{b}_{25}&\vec{b}_{29}&\vec{b}_{32}&\vec{b}_{34}\\
\vec{b}_{10}&\vec{b}_{14}&\vec{b}_{19}&\vec{b}_{24}&\vec{b}_{28}&\vec{b}_{31}\\
\vec{b}_6&\vec{b}_9&\vec{b}_{13}&\vec{b}_{18}&\vec{b}_{23}&\vec{b}_{27}\\
\vec{b}_3&\vec{b}_5&\vec{b}_8&\vec{b}_{12}&\vec{b}_{17}&\vec{b}_{22}\\
\vec{b}_1&\vec{b}_2&\vec{b}_4&\vec{b}_7&\vec{b}_{11}&\vec{b}_{16}
\end{array}
$$
\caption{The enumeration of ${\mathcal{B}}$ in Example~\ref{ex1}}
\label{fig0}
\end{figure}
As an example we obtain $\bar{\rho}({\mbox{\textnormal{ev}}}(2X^5Y^4+5X^3Y^2+4+I))=34$.\demo
\end{example}

\noindent Recall, that the component wise product of two vectors in
${\mathbb{F}}_q^n$ is given by $$(\alpha_1, \ldots , \alpha_n) \ast (\beta_1, \ldots , \beta_n)=(\alpha_1\beta_1,
\ldots , \alpha_n\beta_n).$$
Using this product we can now introduce the concept of
one-way well-behaving pairs. 
\begin{definition}
An ordered pair $(i,j) \in {\mathcal{I}} \times {\mathcal{I}}$ is 
said to be one-way well-behaving (OWB) if $\bar{\rho}(\vec{b}_{i^\prime} \ast
\vec{b}_j)< \bar{\rho}(\vec{b}_{i} \ast \vec{b}_j)$ holds true for all
$i^\prime \in {\mathcal{I}}$ with 
$i^\prime < i$.  
\end{definition}
\begin{example}\label{ex2}
This is a continuation of Example~\ref{ex1}. Consider
$\vec{b}_i={\mbox{\textnormal{ev}}}(X^{\alpha}Y^{\beta}+I)$ and
$\vec{b}_j={\mbox{\textnormal{ev}}}(X^{\gamma}Y^{\delta}+I)$  (where, by assumption, $0 \leq
\alpha , \beta , \gamma , \delta < 6$). The pair $(i,j)$ is OWB if and only if $\alpha+\gamma < 6$ and
$\beta+\delta < 6$ hold simultaneously. To see the ``if'' part note
that if $X^{\eta} Y^{\lambda} \prec_{\deg} X^{\alpha} Y^{\beta}$ then the
leading monomial $M$ of the remainder of
$X^{\eta+\gamma}Y^{\lambda+\delta}$ after division with
$\{X^6-1,Y^6-1\}$ satisfies $M \preceq_{\deg} X^{\eta +\gamma}Y^{\lambda
  +\delta} \prec_{\deg} X^{\alpha+\gamma}Y^{\beta+\delta}$ which follows
from the properties of a monomial ordering and those of the division
algorithm. The ``only if'' part has to do with the special form of the
ideal $I$ coming from a variety which is a Cartesian product. If for
instance, $\alpha+\gamma \geq 6$ then letting $\eta=6-\gamma-1$ we
obtain $X^{\eta}Y^{\beta} \prec_{\deg} X^{\alpha} Y^{\beta}$, but the leading
monomial $N$ of $X^{\eta +\gamma}Y^{\beta +  \delta}$ after division
with $\{X^6-1, Y^6-1\}$ has the same $Y$-part as the leading monomial
of $X^{\alpha +\gamma}Y^{\beta+\delta}$ after division with $\{X^6-1,
Y^6-1\}$, but higher $X$-part. \demo
\end{example}
To formulate the Feng-Rao bound for
primary codes we shall need the following set.
\begin{definition}\label{defdif}
For $i \in {\mathcal{I}}$  define 
$$\Lambda_i=\{ l \in {\mathcal{I}}
\mid \exists \, j \in {\mathcal{I}} {\mbox{ such that }} (i,j) {\mbox{
    is OWB and }} \bar{\rho}(\vec{b}_i \ast \vec{b}_j)=l \}.$$ 
\end{definition}
\noindent Given a $v$-dimensional vector space $U \subseteq {\mathbb{F}}_q^n$ then $\bar{\rho}(U \backslash \{
\vec{0} \} )$ is of size $v$. The following
proposition, known as the Feng-Rao bound for primary codes~\cite[Proposition 8]{geil2014relative}, therefore is operational.
\begin{proposition}\label{prothe}\label{propthe}
Let $U \subseteq {\mathbb{F}}_q^n$ be a vector space of dimension at
least $1$. The support size of $U$ satisfies
\begin{equation}
\# {\mbox{\textnormal{Supp }}} (U) \geq \# \cup_{i \in \bar{\rho}(U \backslash \{
    \vec{0} \})} \Lambda_i. \label{eqbound}
\end{equation}
\end{proposition}
\begin{example}\label{ex3}
This is a continuation of the previous examples. We have
$\vec{b}_{28}={\mbox{\textnormal{ev}}}(X^4Y^3+I)$ and therefore $$\# \Lambda_{28}=\# \{
X^4Y^3, X^5Y^3, X^4Y^4,X^5Y^4,X^4Y^5,X^5Y^5\}=6.$$ In general, for
$\vec{b}_i={\mbox{\textnormal{ev}}}(X^{\alpha}Y^{\beta}+I)$ (with $0 \leq \alpha, \beta <6$) we
have $\# \Lambda_i =(6-\alpha)(6-\beta)$. The situation is depicted in Figure~\ref{fig0.5}.

\begin{figure}[h]
$$
\begin{array}{cccccc}
6&5&4&3&2&1\\
12&10&8&6&4&2\\
18&15&12&9&6&3\\
24&20&16&12&8&4\\
30&25&20&15&10&5\\
36&30&24&18&12&6
\end{array}
$$
\caption{$\# \Lambda_i$ from Example~\ref{ex1} (enumeration with respect to Figure~\ref{fig0})}
\label{fig0.5}
\end{figure}
Let $F(X,Y)$ be a polynomial whose leading monomial with respect to $\prec_{\deg}$ is $X^{\alpha}Y^{\beta}$ for some $0 \leq
\alpha , \beta <6$. Consider $\vec{c}={\mbox{\textnormal{ev}}}(F+I)$, then, by Proposition~\ref{propthe}, $w_H(\vec{c}) \geq (6-\alpha)(6-\beta)$. In general,
Figure~\ref{fig0.5} gives upper bounds on the Hamming weights of all possible words in ${\mathbb{F}}_7^{36}$.\demo
\end{example}

\noindent We now turn to relative generalized Hamming weights. Note that although $C_2 \subset C_1$ implies
 $\bar{\rho}(C_2) \subset \bar{\rho}(C_1)$,
it does not always hold that 
$\vec{c} \in C_1 \backslash C_2$ implies $\bar{\rho}(\vec{c}) \in \bar{\rho}(C_1)\backslash
\bar{\rho}(C_2)$. Nevertheless, the Feng-Rao bound for primary
codes~\cite[Theorem 9]{geil2014relative}  still gives us useful information.
\begin{theorem}\label{thethat}
Consider two linear codes $C_2 \subset C_1\subseteq {\mathbb{F}}_q^n$
with $\dim (C_1)=k_1$ and $\dim (C_2)=k_2$. Let $u$ be the smallest element in $\bar{\rho}(C_1)$ that is not in
  $\bar{\rho}(C_2)$. For $v=1, \ldots , \ell = k_1 - k_2$ we have
\begin{eqnarray}
M_v(C_1,C_2)&\geq &\min \big\{ \# \cup_{s=1}^v \Lambda_{i_s} \mid u\leq
i_1 < \cdots < i_v {\mbox{ and }}
i_1, \ldots ,
i_v \in \bar{\rho}(C_1 \backslash \{\vec{0}\})\big\}. \nonumber
\end{eqnarray}
\end{theorem}
\begin{example}\label{ex5}
This is a continuation of the previous examples. If
$C_2={\mbox{Span}}\{\vec{b}_1, \vec{b}_2,\vec{b}_3, \vec{b}_5\}$ and
$C_1={\mbox{Span}}\{\vec{b}_1, \vec{b}_2,\vec{b}_3, \vec{b}_4,
\vec{b}_5\}$ then $$M_1(C_1,C_2) \geq \min \{ \# \Lambda_4, \#
\Lambda_5 \}=\min \{ 24, 25 \} = 24.$$ However, if 
$C_2={\mbox{Span}}\{\vec{b}_1, \vec{b}_2,\vec{b}_3, \vec{b}_4\}$, while
$C_1$ is unchanged, then $M_1(C_1,C_2) \geq \# \Lambda_5=25$.\demo
\end{example}

To treat dual codes we shall need the following definitions, where the
first can be considered as the counterpart of Definition~\ref{defdif},
and the last as the counterpart of Definition~\ref{defro}.
\begin{definition} For $l \in {\mathcal{I}}$ define
$$V_l=\{i \in {\mathcal{I}} \mid \bar{\rho}(\vec{b}_i \ast \vec{b}_j)=l
  {\mbox{ for some }}\vec{b}_j \in {\mathcal{B}} {\mbox{ with }} (i,j)
  {\mbox{ OWB}}\}.$$
\end{definition}

\begin{definition}
For $\vec{c} \in {\mathbb{F}}_q^n \backslash \{ \vec{0}\}$ define
$\eta(\vec{c})$ to be the smallest number $l \in {\mathcal{I}}$ such that
$\vec{c} \cdot \vec{b}_l \neq 0$. Here $\vec{a} \cdot \vec{b}$ means the
Euclidean inner product between $\vec{a}$ and $\vec{b}$.
\end{definition}
Given a $v$-dimensional space $U$, $\eta(U\backslash \{ \vec{0} \})$ is
of size $v$. The following proposition, known as the
Feng-Rao bound for dual codes~\cite[Proposition 13]{geil2014relative}, therefore is operational.

\begin{proposition}\label{thenaada}
Let $U \subseteq {\mathbb{F}}_q^n$ be a space of dimension at least
$1$. We have 
$$\# {\mbox{\textnormal{Supp}}} (U) \geq \# \cup_{l \in \eta (U\backslash \{\vec{0}\})} V_l.$$
\end{proposition}

\begin{example}\label{ex6}
This is a continuation of the previous examples. For
$\vec{b}_l={\mbox{\textnormal{ev}}}(X^{\alpha}Y^{\beta}+I)$ (with $0 \leq \alpha , \beta  < 6$)
we have $\# V_l=(\alpha+1)(\beta+1)$. Given $\vec{c}$ with
$\eta(\vec{c}) =l$, from Proposition~\ref{thenaada} we obtain
$w_H(\vec{c}) \geq (\alpha +1)(\beta +1)$. By Proposition~\ref{thenaada}, Figure~\ref{fig0.75} gives
upper bounds on the Hamming weights of all possible words in
${\mathbb{F}}_7^{36}$. 
\begin{figure}[h]
$$
\begin{array}{cccccc}
6&12&18&24&30&36\\
5&10&15&20&25&30\\
4&8&12&16&20&24\\
3&6&9&12&15&18\\
2&4&6&8&10&12\\
1&2&3&4&5&6
\end{array}
$$
\caption{$\# V_l$ from Example~\ref{ex5} (enumeration with respect to Figure~\ref{fig0})}
\label{fig0.75}
\end{figure}
\ \demo
\end{example}
We next treat relative generalized Hamming weights. Note that for $C_2 \subset C_1$ it does not in general hold that $\vec{c} \in C_2^\perp \backslash C_1^\perp$
implies $\eta (\vec{c}) \in \bar{\rho} (C_1\backslash \{ \vec{0}\})$. Nevertheless, the
Feng-Rao bound for dual codes~\cite[Theorem 14]{geil2014relative} still gives us useful information.

\begin{theorem}\label{theny}
Consider linear codes $C_2 \subset C_1\subseteq {\mathbb{F}}_q^n$. Let $u^\perp$ be the largest
element in $\bar{\rho}(C_1 \backslash \{\vec{0}\})$. For $v=1, \ldots
, \dim (C_1)-\dim (C_2)=\dim (C_2^\perp) -\dim (C_1^\perp)$ we have
\begin{eqnarray}
M_v(C_2^\perp,  C_1^\perp) &\geq &\min \{ \#  \cup_{s=1}^v V_{i_s}
\mid 1 \leq i_1 < \cdots <i_v \leq u^\perp,  i_1, \ldots , i_v \notin \bar{\rho}(C_2)\}.\nonumber
\end{eqnarray}
\end{theorem}

\begin{example}\label{ex7}
This is a continuation of the previous examples. If
$C_2={\mbox{Span}}\{\vec{b}_1, \vec{b}_2,\vec{b}_3,
\vec{b}_4\}$ and
$C_1={\mbox{Span}}\{\vec{b}_1, \vec{b}_2,\vec{b}_3, \vec{b}_4,
\vec{b}_6\}$, then $$M_1(C_2^\perp,C_1^\perp) \geq \min \{ \# V_5, \#
V_6 \}=\min \{ 4, 3 \} = 3.$$ However, if 
$C_1={\mbox{Span}}\{\vec{b}_1, \vec{b}_2,\vec{b}_3, \vec{b}_4,\vec{b}_5\}$ while
$C_2$ is unchanged then $M_1(C_2^\perp,C_1^\perp) \geq \# V_5=4$.\demo
\end{example}
\ \\

The above theorems and examples lead us to consider the following
family of codes, which have a good behavior with respect to the
applications described in the introduction.  Consider a Cartesian product $S=S_1\times \cdots \times S_m \subseteq
{\mathbb{F}}_q^m$. For $i=1, \ldots , m$ define the one-variable polynomial
\begin{eqnarray}
F_i(X_i)=\prod_{\alpha \in S_i}(X_i-\alpha), \label{eqfi}
\end{eqnarray}
and consider the
vanishing ideal of $S$,  
$I=\langle F_1(X_1), \ldots , F_m(X_m)\rangle \subset
{\mathbb{F}}_q[X_1, \ldots , X_m]$. 
 We write $R={\mathbb{F}}_q[X_1, \ldots , X_m]/I$ and enumerate
$S$ as $S=\{\alpha_1, \ldots  \alpha_n\}$, where $n=\#
S=\prod_{i=1}^ms_i$. Here, we use the notation $s_i=\# S_i$. As in
the above examples we then obtain a vector space homomorphism ${\mbox{\textnormal{ev}}}:R \rightarrow
{\mathbb{F}}_q^n$ defined by ${\mbox{\textnormal{ev}}}(F+I)=(F(\alpha_1), \ldots ,
F(\alpha_n))$. Now let 
\begin{eqnarray}
\Delta(s_1, \ldots , s_m)&=&\{X_1^{i_1} \cdots X_m^{i_m} \mid 0 \leq i_t <s_t, t=1,
          \ldots , m\} = \{N_1, \ldots , N_n\}, \nonumber 
\end{eqnarray}
where the enumeration of the $N_i$'s is with respect to an arbitrary
(but fixed) monomial ordering $\prec$. For general $L \subseteq \Delta(s_1, \ldots , s_m)$
define 
$$C(L)={\mbox{Span}} \{ {\mbox{\textnormal{ev}}}(X_1^{i_1} \cdots X_m^{i_m}+I) \mid
X_1^{i_1}\cdots X_m^{i_m} \in L\},$$
which is clearly a code of length $n$. 
For the purpose of applying the Feng-Rao bounds to the codes $C(L)$ and $C(L)^\perp$ we 
introduce 
 functions $D$ and $D^\perp$.
\begin{definition}
Given $X_1^{i_1} \cdots X_m^{i_m}\in \Delta(s_1,\ldots , s_m)$,   
define
$$D(X_1^{i_1} \cdots X_m^{i_m})=\prod_{t=1}^{m}(s_t-i_t) {\mbox{ and }}
D^\perp(X_1^{i_1} \cdots
X_m^{i_m})=\prod_{t=1}^m(i_t+1).$$ More generally, for $K
\subseteq \Delta(s_1, \ldots , s_m)$ let
\begin{eqnarray}
D(K)&=&\# \{ N \in \Delta(s_1, \ldots , s_m) \mid N {\mbox{ is divisible by some }} M \in
        K\}, \nonumber \\
D^\perp(K)&=&\# \{N \in \Delta (s_1, \ldots , s_m) \mid N {\mbox{ divides
                              some }} M \in K\}.\nonumber
\end{eqnarray}
\end{definition}

Observe that $D(X_1^{i_1}, \ldots , X_m^{i_m})=D(\{X_1^{i_1}, \ldots ,
X_m^{i_m}\})$ and $D^\perp(X_1^{i_1}, \ldots ,
X_m^{i_m})=D^\perp(\{ X_1^{i_1}, \ldots ,
X_m^{i_m}\})$.

We are now ready to describe the relative parameters of the evaluation
codes introduced above.

\begin{theorem}\label{theour}
Consider $S=S_1 \times \cdots \times S_m
\subseteq {\mathbb{F}}_q^m$ and $L_2 \subset L_1 \subseteq \Delta
(s_1, \ldots , s_m)=\{N_1, \ldots ,
N_n\}$ where the enumeration of the $N_i$'s is with respect to
an arbitrary (but fixed) monomial ordering $\prec$. Then the codes $C(L_1)$ and $C(L_2)$ are of length $n$ and
their codimension equals
$\#L_1-\#L_2$. Furthermore, for $v=1, \ldots , \# L_1 - \# L_2$ we
have
\begin{eqnarray}
&M_v(C(L_1),C(L_2))\geq  \min \{{\mbox{$D$}}(K) \mid K \subseteq \{ N_u,
                          \ldots , N_n\} \cap L_1, \#K=v\},
                          \label{bnd}\\
&M_v(C(L_2)^\perp,C(L_1)^\perp)\geq  \min \{
                                       D^\perp(K)
                                      \mid K \subseteq \{N_1, \ldots
                                      N_{u^\perp} \} \backslash L_2, 
                                      \#K=v   \}, \label{bndperp} \
\end{eqnarray}
where $u=\min \{i \mid N_i \in L_1 \backslash L_2\}$ and $u^\perp =
\max \{i \mid N_i \in L_1 \}$.
\end{theorem}
\begin{proof}
We start by proving that $\{{\mbox{\textnormal{ev}}}(N_1+I),\ldots ,
{\mbox{\textnormal{ev}}}(N_n+I)\}$ is a basis for ${\mathbb{F}}_q^n$ as a
vector space over ${\mathbb{F}}_q$. This fact implies that the dimension of
$C(L_i)$ equals $\#L_i, i=1,2$, and the formula for the
codimension follows. Observe that
\begin{eqnarray}
{\mbox{\textnormal{ev}}}(F(X_1, \ldots , X_m)+I)={\mbox{\textnormal{ev}}}\big(F(X_1, \ldots , X_m)
-A_1(X_1, \ldots , X_m)F_1(X_1)-\cdots\nonumber \\
 {\mbox{ \ \ \ \ \ \ \ \ \ \ \ \ \ \ \ \ \ \ \ \ }} - A_m(X_1, \ldots ,
X_m)F_m(X_1, \ldots , X_m)+I\big)\nonumber 
\end{eqnarray}
for any polynomials $A_1, \ldots , A_m$ in the variables $X_1, \ldots
X_m$ with
coefficients in ${\mathbb{F}}_q$. Hence, we
may assume that $\deg_{X_1} F < \deg F_1=s_1,
\ldots$, $\deg_{X_m} F < \deg F_m = s_m$. Using Lagrange-interpolation
we next see that ${\mbox{\textnormal{ev}}}$ is surjective and, as $\Delta(s_1, \ldots , s_m)$ is of the same
size as the image ${\mathbb{F}}_q^n$, the considered set is indeed a
basis for ${\mathbb{F}}_q^n$. As in the above examples we enumerate
the basis elements $\{\vec{b}_1={\mbox{\textnormal{ev}}}(N_1+I), \ldots ,
\vec{b}_n={\mbox{\textnormal{ev}}}(N_n+I) \}$ (meaning that we order it according
to the monomial ordering $\prec$). \\
Notice that, regardless of the choice of monomial ordering $\prec$,
$D(N_i) \leq \# \Lambda_i$, 
and that in larger generality\footnote{Actually equalities hold -- which
  can be seen by applying similar arguments as in Example~\ref{ex2} --
  but we shall not need this fact.}  $D(\{N_{i_1}, \ldots , N_{i_m}\}) \leq \#
\cup_{t=1}^m \Lambda_{i_t}$. 
Therefore (\ref{bnd})  follows from
Theorem~\ref{thethat}. Similarly, regardless of the choice of monomial
ordering $\prec$, $D^\perp(N_i) \leq \# V_i$ and in larger
generality\footnote{Actually again equalities hold, but we shall
not need this fact.}
 $D^\perp(\{N_{i_1}, \ldots , N_{i_m}\}) \leq \#
\cup_{t=1}^m V_{i_t}$. Hence, (\ref{bndperp}) follows from Theorem~\ref{theny}.
\end{proof}
In the next 
two sections we show  a way to choose $L_2 \subset L_1$ such that the parameters 
$\ell=\dim C(L_1)-\dim C(L_2)=\# L_1-\# L_2$, $M_1(C(L_1),C(L_2))$,
and $M_1(C(L_2)^\perp,C(L_1)^\perp)$ are good. 
We study two separate cases. The first case, which we treat in
Section~\ref{seclarge}, deals with relatively large codimension. The second case,
which we treat in Section~\ref{secsmall}, concerns relatively small codimensions.

Before studying these two families of codes we briefly discuss the
decoding of the related asymmetric quantum codes. Observe that the decoding algorithm for order domain codes described
in~\cite[Section 6.3]{handbook} can be applied in the more general
setting of linear codes with Feng-Rao designed minimum
distance. Hence, it can be applied to all codes of the present paper. This holds both for dual
codes~\cite{MM} and primary codes~\cite{agismm}. The decoding algorithm which
corrects errors up to half the designed minimum distance uses
${\mathcal{O}}(n^3)$ operations, where $n$ is the length of the
code. In~\cite[Appendix A]{MR2588125} Duursma and Park provided a similar
algorithm correcting errors up to half the designed {\it{relative}}
minimum distance. This was done at the general level of linear codes
described by means of their parity check matrix. The
application in connection with  decoding
of asymmetric quantum codes is as follows. To decode
both phase-shift and qudit-flip errors up to half the designed values
of $d_z$ and $d_x$ one
will need two decoding algorithms, namely one which decodes up to
$\lfloor (M_1(C_1,C_2)-1)/2\rfloor$ errors in connection with $C_2
\subset C_1$, and another which corrects up to $\lfloor (M_1(C_2^\perp,C_1^\perp)-1)/2\rfloor$ errors in connection with $C_1^\perp
\subset C_2^\perp$~\cite{q8,q7}. Duursma and Park's algorithm applies to the last
task, but not to the first
 in its present form. It is an open research problem to modify 
the decoding algorithm from~\cite{MR2588125} for
 general nested dual codes, so that it also works for nested
 {\it{primary}} codes. This probably could be done using the material in~\cite{agismm}. 
In the absence of such a
translation one may instead apply the algorithm from~\cite{agismm} to
correct only up to $\lfloor (d(C_1)-1)/1\rfloor$ errors in connection
with the primary nested codes.

\section{Relatively large codimension}\label{seclarge}
One of the nice features of the Feng-Rao bounds is that they come with
improved code constructions. In the setting of the codes in the
previous section, by applying the improved construction for primary codes~\cite{AG}, we obtain  a code $C(L_1)$ of designed
distance $\delta$ and maximal dimension if we choose
\begin{eqnarray}
L_1=\{X_1^{i_1} \cdots X_m^{i_m} \in \Delta(s_1, \ldots , s_m) \mid D(X_1^{i_1}\cdots X_m^{i_m})
  \geq \delta\}. \label{eqsnabeleen}
\end{eqnarray}
Similarly, by applying the improved construction for dual codes~\cite{FR1,FR2} we obtain a code $C(L_2)^\perp$ of designed distance
$\delta^\perp$ and maximal dimension if we choose
\begin{eqnarray}
L_2=\{X_1^{i_1} \cdots X_m^{i_m} \in \Delta(s_1, \ldots , s_m) \mid D^\perp(X_1^{i_1}\cdots X_m^{i_m})
  < \delta^\perp \}. \label{eqsnabeltoo}
\end{eqnarray}
Our first proposal for constructing good pairs of nested codes
is to choose $L_1$ and $L_2$ as in
(\ref{eqsnabeleen}) and (\ref{eqsnabeltoo}) with $L_2 \subset
L_1$. We then obtain
\begin{eqnarray}
M_1(C(L_1),C(L_2)) &\geq d(C(L_1)) &\geq \delta ,\label{eq1}\\
M_1(C(L_2)^\perp , C(L_1)^\perp) &\geq d(C(L_2)^\perp)&\geq
  \delta^\perp.\label{eq2}
\end{eqnarray}
The codimension $\ell= \#L_1-\#L_2$ is the largest possible with these
designed parameters as, by Proposition~\ref{prothe} and
Proposition~\ref{thenaada}, 
$L_1$ is as large as possible and $L_2$ is as small as possible, such that
(\ref{eq1}) and (\ref{eq2}) hold. Observe that (\ref{eq1}) and
(\ref{eq2}) are independent of the choice of monomial
ordering $\prec$ in Theorem~\ref{theour}, as the integers $u$ and
$u^\prime$ from that theorem play no role here. 
Note that $D(X_1^{i_1} \cdots X_m^{i_m}) =\delta$ is
a concave function on the domain under consideration, while
$D^\perp(X_1^{i_1} \cdots X_m^{i_m}) =\delta^\perp$ is a convex
function. Therefore the necessary condition that $L_2 \subset L_1$
creates a restriction on how small a codimension can be  for
each fixed value of $\delta$ (low codimensions require another method which we
describe in the next section).  The following theorem summarizes the
method described.
\begin{theorem}\label{thebigB}
With the above notation, fix two positive integers $\delta$ and
$\delta^\perp$ such that the monomial sets $L_1$ and $L_2$ described
in~(\ref{eqsnabeleen}) and (\ref{eqsnabeltoo}), respectively, satisfy
that $L_2 \subset L_1$. Then the evaluation codes $C(L_2)\subset
C(L_1)$ are of codimension $\ell=\#L_1-\#L_2$, and the relative minimum
distances satisfy $M_1(C(L_1),C(L_2))\geq
\delta$ and $M_1(C(L_2)^\perp,C(L_1)^\perp)\geq \delta^\perp$.
\end{theorem}

\begin{remark}
The case $\delta=\delta^\perp$ and $S={\mathbb{F}}_q^m$ is studied in~\cite{galindo2015new,galindo2015stabilizer} in
connection with symmetric quantum codes. There it is characterized
when the corresponding sets $L_1$ and $L_2$ satisfy the inclusion $L_2
\subset L_1$. 
\end{remark}
\begin{remark}
It is possible to show that
$d(C(L_1))=M_1(C(L_1),C(L_2))$, but for general Cartesian product
point sets it is unsettled if
$$d(C(L_2)^\perp,C(L_1)^\perp)=M_1(C(L_2)^\perp,C(L_1)^\perp),$$ leaving
it undecided if the related asymmetric quantum codes are pure or not.
\end{remark}

Below we analyze the parameters of the nested code pairs in
Theorem~\ref{thebigB} when the sets $S_1, \ldots , S_m$ are all of the
same size, but first we illustrate the theorem with an example.

\begin{example}\label{ex8}
This is a continuation of the examples in Section~\ref{sec4} where we
considered $S={\mathbb{F}}_7^\ast \ast {\mathbb{F}}_7^\ast$. From
Figure~\ref{fig0.5} we see that for $\delta=12$ the set $L_1$ in (\ref{eqsnabeleen})
becomes
\begin{eqnarray}
L_1&=&\{ 1, X, Y, X^2,XY,Y^2,
X^3,X^2Y,XY^2,Y^3,X^4,\nonumber \\
&& {\mbox{ \ {\hspace{2.3cm}} \ }}X^3Y,X^2Y^2,XY^3,Y^4, X^3Y^2,X^2Y^3\}.\label{eqypti1}
\end{eqnarray}

Hence, $\# L_1=17$. According to Figure~\ref{fig0.75}, $\delta^\perp =
6 $ is the highest possible value of $\delta^\perp$ such
that all $X^\alpha Y^\beta \in \Delta(6,6)$ with $D^\perp (X^\alpha
Y^\beta ) < \delta^\perp$ also belong to $L_1$. The
corresponding set $L_2$ in (\ref{eqsnabeltoo}) then becomes
\begin{eqnarray}
L_2=\{ 1, X,Y,X^2,XY,Y^2,X^3,Y^3,X^4,Y^4\}\label{eqypti2}
\end{eqnarray}
which is of size $10$. Hence, the codimension between $C(L_1)$ and
$C(L_2)$ is $7$, and $M_1(C(L_1),C(L_2))\geq 12$ and $M_1(C(L_2)^\perp,
C(L_1)^\perp)\geq 6$. In a similar fashion we obtain the remaining parameters in
Table~\ref{tabwone}. Note that if we have a code pair $C(L_2) \subset
C(L_1)$ with designed parameters $\delta=a$, $\delta^\perp=b$ and $\#
L_1- \# L_2=\ell$
then there exist $L_2^\prime \subset L_1^\prime$ with 
designed parameters of $C(L_2^\prime) \subset C(L_1^\prime)$  being $\delta=b$,
$\delta^\perp =a$ and $\# L_1^\prime - \# L_2^\prime=\ell$. Hence, in
Table~\ref{tabwone} we only list parameters with $\delta \geq
\delta^\perp$. We also exclude cases with $\delta^\perp=1$
(corresponding to $C_2=\{\vec{0}\}$).
\begin{table}
\begin{center}
\begin{tabular}{c|rrrrrrrrrrrrrrr}
$\ell$&2&1&3&1&3&5&3&5&7&2&5&7&9\\
$\delta$&30&25&25&24&24&24&20&20&20&18&18&18&18\\
$\delta^\perp$&2&3&2&4&3&2&4&3&2&5&4&3&2\\
\ \\
$\ell$&3&6&8&10&5&8&10&12&7&9&12&14&16\\
$\delta$&16&16&16&16&15&15&15&15&12&12&12&12&12 \\
$\delta^\perp$&5&4&3&2&5&4&3&2&6&5&4&3&2 \\
\ \\
$\ell$&9&11&14&16&18&10&12&15&17&19&12&14&17\\
$\delta$&10&10&10&10&10&9&9&9&9&9&8&8&8 \\
$\delta^\perp$&6&5&4&3&2&6&5&4&3&2&6&5&4 \\
\ \\
$\ell$&19&21&16&18&21& 23&25&20&23&25&27&26&28 \\
$\delta$&8&8&6&6&6&6&6&5&5&5&5&4&4 \\
$\delta^\perp$&3&2&6&5&4& 3&2&5&4&3&2&4&3 \\
\ \\
$\ell$ &30&30&32&34 \\
$\delta$&4&3&3&2    \\     
$\delta^\perp$&2&3&2&2   
\end{tabular}
\end{center}
\caption{Parameters from Example~\ref{ex8}.}
\label{tabwone}
\end{table}
\ \demo
\end{example}
In the following we find closed formula expressions for the
parameters of the coset construction in Theorem~\ref{thebigB} when
$S_1, \ldots , S_m$ are all of the same size. We start with a lemma
explaining for which choices of $\delta$ and $\delta^\perp$ the theorem 
works.
\begin{lemma}\label{lembetingelse}
Assume $s=s_1=\cdots =s_m$ and consider $\delta \in \{1, \ldots ,
s^m\}$. Let $v \in \{ 0, \ldots , m-1 \}$ be such that $s^v \leq
\delta \leq s^{v+1}$. If $\delta^\perp \leq \lfloor 
(s-\frac{\delta}{s^v}+1)s^{m-v-1}\rfloor$ then the set $L_2$
from~(\ref{eqsnabeltoo}) is contained in the set $L_1$ from~(\ref{eqsnabeleen}).
\end{lemma}
\begin{proof}
Define functions $\tilde{D}:{\mathbb{Q}}^m \rightarrow {\mathbb{Q}}$
and $\tilde{D}^\perp:{\mathbb{Q}}^m \rightarrow {\mathbb{Q}}$ by
$\tilde{D}((i_1, \ldots , i_m))=\prod_{t=1}^m (s-i_t)$ and
$\tilde{D}^\perp((i_1, \ldots , i_m))=\prod_{t=1}^m (i_t+1)$. Let
$i=s-\delta/s^v$ and note that $0 \leq i \leq s-1$ as well as 
$$\tilde{D} ((\underset{v {\mbox{ {\scriptsize{times}} }}}{\underbrace{{0, \ldots ,
  0}}},i,\underset{m-1-v {\mbox{ {\scriptsize{times}} }}}{
\underbrace{s-1,\ldots ,s-1}}))= \delta$$
hold true.
Finally we observe that
$$\tilde{D}^\perp ((\underset{v {\mbox{ {\scriptsize{times}} }}}{\underbrace{{0, \ldots ,
  0}}},i,\underset{m-1-v {\mbox{ {\scriptsize{times}} }}}{
\underbrace{s-1,\ldots ,s-1}}))=(s-\frac{\delta}{s^v}+1)s^{m-v-1}$$
and the lemma follows.
\end{proof}
The next step in our analysis is to establish an estimate from below on the dimension
of the code $C(L_1)$ and $C(L_2)^\perp$ when $L_1$ is as
in~(\ref{eqsnabeleen}), $L_2$ is as in~(\ref{eqsnabeltoo}) and
$s=s_1=\cdots =s_m$. In~\cite[Theorem 1]{hyptype} a bound was given for
the special case $S_1=\cdots =S_m={\mathbb{F}}_q$ and $q^{m-1} \leq
\delta, \delta^\perp \leq q^m$. With the last mentioned condition we
do not obtain
$L_2 \subset L_1$ and we therefore now generalize the result
from~\cite{hyptype} to arbitrary $1 \leq \delta^\perp , \delta \leq
s^m$ and $s=s_1=\cdots =s_m$. We start with a technical lemma, whose 
proof we give in Appendix~\ref{apa}.

\begin{lemma}\label{lemuseful}
For $m \geq 2$, $1 \leq i \leq m$ it holds that
\begin{eqnarray}
\int_0^{s-\frac{\tau}{s^{m-i}(s-x_1)\cdots (s-x_{i-1})}}
  \int_0^{s-\frac{\tau}{s^{m-i-1}(s-x_1)\cdots (s-x_{i})}}\cdots
  \int_0^{s-\frac{\tau}{(s-x_1)\cdots (s-x_{m-1})}}  dx_m  \cdots
  dx_{i+1} \, dx_i \nonumber \\
= s^{m-i+1}-\sum_{t=0}^{m-i}\frac{1}{t
  !}\frac{\tau}{(s-x_1) \cdots (s-x_{i-1})}\bigg( \ln \bigg(
  \frac{(s-x_1) \cdots (s-x_{i-1})s^{m-i+1}}{\tau}\bigg) \bigg)^t.\nonumber
\end{eqnarray}
\end{lemma}

From this lemma we obtain information on the dimensions of the codes
as follows.

\begin{theorem}\label{thedimensions}
Let $s=s_1=\cdots =s_m$ and consider $L_1$ and $L_2$ as
in~(\ref{eqsnabeleen}) and (\ref{eqsnabeltoo}), respectively with
$\delta=\delta^\perp=\tau \in \{1, \ldots , s^m\}$. The dimensions of
$C(L_1)$ and $C(L_2)^\perp$ are at least
\begin{eqnarray}
s^m-\sum_{t=1}^m\frac{1}{(t-1)!}\tau \bigg( \ln \big( \frac{s^m}{\tau} \big)\bigg)^{t-1}.\label{eqgenerel}
\end{eqnarray} 
If $1 \leq \tau < s$ then the dimensions are at least
\begin{eqnarray}
s^m-\sum_{t=1}^m\frac{1}{(t-1)!}\tau((m-1)\ln (\tau))^{t-1}\label{eqspecifik}
\end{eqnarray}
which is sharper than~(\ref{eqgenerel}).
\end{theorem}
\begin{proof}
By symmetry is is enough to prove the result for $C(L_1)$. The
dimension of $C(L_1)$, i.e.\ the number of integer tupples $(i_1,
\ldots , i_m) \in \{0, \ldots , s-1\}^m$ with $(s-i_1) \cdots (s-i_m)
\geq \tau$, is at least that of the volume of 
$$\{(x_1, \ldots , x_m) \in [0,s]^m \mid (s-x_1)\cdots (s-x_m) \geq
\tau\}$$
which corresponds to the integral in Lemma~\ref{lemuseful} when $i$ is
chosen to be equal to $1$. This proves~(\ref{eqgenerel}). Next, assume
$1 \leq \tau < s$. The above mentioned set of integer tupples can be
divided into two sets, the first set consisting of those tupples satisfying $0
\leq i_v<s-\tau$ for some $v \in \{1, \ldots , m\}$, and the second
set consisting of those tupples satisfying $s-\tau\leq i_v$ for $v=1, \ldots ,
m$. The number of elements in the first set equals $s^m-\tau^m$. The
cardinality of the second set is estimated from below by the volume of
$$\{(x_1, \ldots , x_m) \in [0,\tau]^m \mid (\tau-x_1) \cdots
(\tau-x_m) \geq \tau \}.$$
The last part of the theorem now follows by applying
Lemma~\ref{lemuseful} with $i=1$ and $s=\tau$.
\end{proof}

\begin{remark}
From Theorem~\ref{thedimensions} one obtains for each choice of $m$
closed formula lower bounds on the rate $k/n$ as
a function of the relative minimum distance\footnote{Here, the
  relative minimum distance should not be confused with the first
  relative generalized Hamming weight.} $d/n$. Such estimates are
independent of the actual value of $s$. From the proof it is clear
that these estimates become more and more precise as $s$ increases.
Computer experiments reveal that with
$m=2$ and $m=3$ already for $s=32$ the true values of the rate is
almost the same as the estimated.
\end{remark}

Using the constructions described in Section~\ref{secintro} we get by applying
Lemma~\ref{lembetingelse} in combination with (\ref{eqgenerel}) from
Theorem~\ref{thedimensions} the following result on the existence of ramp secret
sharing schemes and asymmetric quantum codes.
\begin{theorem}\label{theappllarge}
Consider integers $m \geq 2$ and $s\leq q$, where $q$ is a prime
power. Given $\delta \in \{1,\ldots , s^m\}$ let $v \in \{0,\ldots ,
m-1\}$ be such that $s^v\leq \delta \leq s^{v+1}$ and choose an
integer $\delta^\perp \leq \lfloor (s-\frac{\delta}{s^v}+1)s^{m-v-1}
\rfloor$. From Theorem~\ref{thebigB} we obtain ramp secret sharing
schemes over ${\mathbb{F}}_q$ with $n=s^m$ participants, shares in
${\mathbb{F}}_q^\ell$ where 
$$\ell \geq s^m-\sum_{t=1}^m \frac{1}{(t-1)!}\bigg( \delta \big(\ln
\big(\frac{s^m}{\delta}\big)\big)^{t-1}+\delta^\perp \big(\ln
\big(\frac{s^m}{\delta^\perp}\big)\big)^{t-1} \bigg),$$
the first privacy number satisfying $t=t_1 \geq \delta^\perp-1$ and
the last reconstruction number satisfying $r=r_\ell \leq s^m-\delta
+1$. Similarly, we obtain asymmetric quantum codes with parameters 
$$[[n=s^m,\ell,d_z \geq \delta / d_x \geq \delta^\perp]]_q.$$
\end{theorem}
\begin{remark}
The lower bound on $\ell$ in Theorem~\ref{theappllarge} can
be improved in the case $1 \leq \delta < s$ or $1 \leq \delta^\perp <
s$ by applying~(\ref{eqspecifik}) instead of (\ref{eqgenerel}). 
\end{remark}

Recall from Remark~\ref{remsammenhaeng} that studying asymmetric
quantum codes derived from the CSS construction is equivalent to
studying linear ramp secret sharing schemes. Therefore, the following
discussion on asymmetric quantum codes imposed by
Theorem~\ref{thebigB} can be directly translated into results on ramp
secret sharing schemes. We leave the details for the reader. It seems relevant to compare in some concrete cases what can be derived from
Theorem~\ref{thebigB} in combination with Theorem~\ref{thecss} with
other general constructions in the literature of asymmetric quantum codes of
similar length. Applying Theorem~\ref{thebigB} to polynomials in two
variables and to a Cartesian product $S=S_1 \times S_2$ we obtain from
Theorem~\ref{thecss} asymmetric stabilizer codes of length $s_1s_2$,
where $s_1=\#S_1$ and $s_2=\#S_2$. For comparison La Guardia's Construction II of asymmetric
quantum generalized Reed-Solomon codes~\cite[Theorem 7.1]{aqlaguardia} 
gives codes of length
$m_1m_2$ as follows:
\begin{theorem}\label{thelaguardia}
Let $q$ be a prime power. Then there exist asymmetric quantum
generalized Reed-Solomon codes with parameters
$$[[m_1 m_2, \ell=m_1(2k-m_2+c),d_z \geq d / d_x \geq (d-c)]]_q,$$
where $1 < k < m_2 < 2k+c\leq q^{m_1}$, $k=m_2-d+1$, and $m_2, d >
c+1$, $c \geq 1$, $m_1 \geq 1$ are integers.
\end{theorem}
Observe that the bound $d_z \geq m_2-k+1$ in
Theorem~\ref{thelaguardia} suggests that to obtain the widest variety
of code parameters for a given code length one should choose $m_1$
smallest possible and $m_2$ largest possible,  such that the conditions in
the theorem are satisfied. The surprising consequence -- which we
illustrate in the following example -- is that sometimes one obtains
better parameters from Theorem~\ref{thelaguardia}  by considering a
shorter code length. Adding 0s to the code words of the shorter code, 
we then obtain bounds on codes of the right length.
\begin{example}\label{exq1}
We first consider asymmetric quantum codes with $q=7$ and of
length $49$. Applying Theorem~\ref{thelaguardia} directly to the case
of $n=49$ we obtain six different
sets of parameters. However, we can actually derive better information
on codes of length $49$ by applying Theorem~\ref{thelaguardia}
to codes of length $n=48$. 
In Table~\ref{tablaguardia1} a selection of such  parameters are
compared with examples of what can be achieved by applying 
Theorem~\ref{thebigB} instead. In the table, by ``-----'' we indicate that no
comparable parameters can be derived. The advantage of our method is clear in most
cases. Furthermore, all the codes in Table~\ref{tablaguardia1}
coming from Theorem~\ref{thebigB} strictly exceed the Gilbert-Varshamov bound (Theorem~\ref{theMGV}).\\
\begin{table}
\begin{center}
\begin{tabular}{cc}
Theorem~\ref{thelaguardia} (\cite[Theorem 7.1]{aqlaguardia}) &Theorem~\ref{thebigB}\\
\hline
----- &$[[49,3,30/4]]_7$\\
----- &$[[49,8,24/4]]_7$\\
-----&$[[49,5,24/5]]_7$\\
-----&$[[49,9,20/5]]_7$\\
$[[49,10,14/7]]_7$&$[[49,10,14/7]]_7$\\
$[[49,12,14/6]]_7$&$[[49,14,14/6]]_7$\\
$[[49,16,12/6]]_7$&$[[49,18,12/6]]_7$\\
$[[49,18,10/7]]_7$&$[[49,16,10/7]]_7$
\end{tabular}
\end{center}
\caption{Comparison of code parameters in Example~\ref{exq1}. 
Parameters from applying Theorem~\ref{thelaguardia} to
  codes of length $n=48$ on the left,  and parameters from Theorem~\ref{thebigB} on the right.}
\label{tablaguardia1}
\end{table}
We next consider asymmetric quantum codes with $q=8$ and  of
length $64$. Our theorem treats many more constellations of $d_z/d_x$
with $d_z \geq d_x$ than does Theorem~\ref{thelaguardia}. For
instance, the highest value of $d_z$ treated by
Theorem~\ref{thelaguardia} is $d_z=31$, whereas Theorem~\ref{thebigB} describes
35 different nested code pairs with $d_z \geq 32$. In most cases the
code parameters guaranteed by Theorem~\ref{thebigB} are much better
than the parameters described in Theorem~\ref{thelaguardia}. However,
there are also cases where the situation is the opposite. From the huge
amount of obtainable values $d_z/d_x$ we display in
Table~\ref{tablaguardia2} some representative examples that illustrate 
the situation. Again all listed codes coming from Theorem~\ref{thebigB} strictly exceed the
Gilbert-Varshamov bound.
\begin{table}
\begin{center}
\begin{tabular}{cc}
Theorem~\ref{thelaguardia} (\cite[Theorem 7.1]{aqlaguardia})&Theorem~\ref{thebigB}\\
\hline
-----&$[[64,5,35/5]]_8$\\
-----&$[[64, 12,30/4]]_8$\\
-----&$[[64,9,30/5]]_8$\\
-----&$[[64,7,30/6]]_8$\\
$[[64,6,25/6]]_8$&$[[64,10,25/6]]_8$\\
$[[64,6, 24/7]]_8$&$[[64,10,24/7]]_8$\\
$[[64,50,5/4]]_8$&$[[64,51,5/4]]_8$
\end{tabular}
\end{center}
\caption{Comparison of code parameters in
  Example~\ref{exq1}. Parameters from Theorem~\ref{thelaguardia} on
  the left and parameters from Theorem~\ref{thebigB} on the right.}
\label{tablaguardia2}
\end{table}
\ \demo
\end{example}

\begin{example}\label{exsyvseks}
In this example we consider $S_1, S_2 \subseteq {\mathbb{F}}_7$ with
$\# S_1=6$ and $\#S_2=7$ and apply Theorem~\ref{thebigB} and
Theorem~\ref{thelaguardia} to construct asymmetric quantum codes of
length $n=42$. In most cases Theorem~\ref{thebigB} is much better than
Theorem~\ref{thelaguardia}, however, there also exists a number of cases
where the latter is the best. Table~\ref{tabsyvseks}
displays some illustrative examples. As in the previous example all
displayed cases coming from Theorem~\ref{thebigB} strictly exceed the Gilbert-Varshamov bound.
\begin{table}
\begin{center}
\begin{tabular}{cc}
Theorem~\ref{thelaguardia} (\cite[Theorem 7.1]{aqlaguardia})&Theorem~\ref{thebigB}\\
\hline
-----&$[[42,4,20/5]]_7$\\
$[[42,2,18/4]]_7$&$[[42,9,18/4]]_7$\\
$[[42,6,16/4]]_7$&$[[42,10,16/4]]_7$\\
$[[42,10,14/4]]_7$&$[[42,13,14/4]]_7$\\
$[[42,14,10/6]]_7$&$[[42,14,10/6]]_7$\\
$[[42,16,9/6]]_7$&$[[42,15,9/6]]_7$\\
$[[42,24,7/4]]_7$&$[[42,23,7/4]]_7$\\
$[[42,28,5/4]]_7$&$[[42,29,5/4]]_7$
\end{tabular}
\end{center}
\caption{Comparison of code parameters in
  Example~\ref{exsyvseks}. Parameters from Theorem~\ref{thelaguardia} on
  the left and parameters from Theorem~\ref{thebigB} on the right.}
\label{tabsyvseks}
\end{table}

\ \demo
\end{example}

\begin{example}
This is a continuation of Example~\ref{ex8}. From Table~\ref{tabwone}
one can show that all related asymmetric quantum codes strictly exceed
the Gilbert-Varshamov bound (Theorem~\ref{theMGV}). The
details are left for the reader.
\end{example}

We conclude this section with an example illustrating how to derive
relative generalized Hamming weights of the considered codes. Recall
from Section~\ref{secintro}
that such information directly translates into information on the privacy numbers and the
reconstruction numbers of the corresponding ramp secret sharing schemes.
\begin{example}\label{ex9}
In this example we apply Theorem~\ref{theour} to estimate the
parameters $M_v(C(L_1),C(L_2))$, $M_v(C(L_2)^\perp, C(L_1)^\perp)$,
$v=1, \ldots , \#L_1-\#L_2=7$ where $L_1$ and $L_2$ are as in (\ref{eqypti1})
and (\ref{eqypti2}), respectively. See Figure~\ref{figfjongdu}.
\begin{figure}[h]
$$
\begin{array}{cccccc}
N_{21}&N_{26}&N_{30}&N_{33}&N_{35}&N_{36}\\
{\text{\rund{$N_{15}$}}}&N_{20}&N_{25}&N_{29}&N_{32}&N_{34}\\
{\text{\rund{$N_{10}$}}}&{\text{\rectangled{$N_{14}$}}}&{\text{\rectangled{$N_{19}$}}}&N_{24}&N_{28}&N_{31}\\
{\text{\rund{$N_6$}}}&{\text{\rectangled{$N_9$}}}&{\text{\rectangled{$N_{13}$}}}&{\text{\rectangled{$N_{18}$}}}&N_{23}&N_{27}\\
{\text{\rund{$N_3$}}}&{\text{\rund{$N_5$}}}&{\text{\rectangled{$N_8$}}}&{\text{\rectangled{$N_{12}$}}}&N_{17}&N_{22}\\
{\text{\rund{$N_1$}}}&{\text{\rund{$N_2$}}}&{\text{\rund{$N_4$}}}&{\text{\rund{$N_7$}}}&{\text{\rund{$N_{11}$}}}&N_{16}
\end{array}
$$
\caption{The situation in Example~\ref{ex9}: $L_2$ corresponds to the
  circled monomials and $L_1$ equals $L_2$ plus the boxed monomials.}
\label{figfjongdu}
\end{figure}
Recall that the monomial ordering, that we use in this example, is the
degree lexicographic ordering $\prec_{\deg}$. We therefore obtain
$u=\min \{ i \mid N_i \in L_1 \backslash L_2 \}=8$
and
$u^\perp=\max \{i \mid N_i \in L_1\}=19$.
Hence, (\ref{bnd}) becomes 
\begin{eqnarray}
M_v(C(L_1),C(L_2)) &\geq& \min \big\{ D(K) \mid K \subseteq \{N_8,
                          N_9, N_{10} \nonumber \\
&& {\mbox{ \ \ \ \ \ }} N_{11}, N_{12}, N_{13}, N_{14}, N_{15},
   N_{18}, N_{19}\}, \#K=v \big\}, \nonumber
\end{eqnarray}
and (\ref{bndperp}) becomes
\begin{eqnarray}
M_v(C(L_2)^\perp,C(L_1)^\perp) &\geq& \min \big\{ D^\perp(K) \mid K \subseteq \{N_8, N_9,
                          N_{12} \nonumber \\
&& {\mbox{ \ \ \ \ \ }} N_{13}, N_{14}, N_{16}, N_{17}, N_{18},
   N_{19}\}, \#K=v \big\}. \nonumber
\end{eqnarray}
Going through all possible combinations we obtain the information
in Table~\ref{tabtob}. Hence, we can construct a ramp secret sharing
scheme over ${\mathbb{F}}_7$ with $n=36$ participants, with the
secrets belonging to ${\mathbb{F}}_7^7$ and with the privacy numbers
being as in Table~\ref{tabpam}.
\begin{table}
\begin{center}
\begin{tabular}{c|rrrrrrr}
$v$&1&2&3&4&5&6&7\\
$M_v(C(L_1),C(L_2))\geq$&12&15&16&18&20&22&23\\
$M_v(C(L_2)^\perp,C(L_1)^\perp)\geq$&6&8&9&11&12&14&15
\end{tabular}
\end{center}
\caption{Estimated relative generalized Hamming weights of the code pair in Example~\ref{ex9}.}
\label{tabtob}
\end{table}

\begin{table}
\begin{center}
\begin{tabular}{c|rrrrrrr}
$v$&1&2&3&4&5&6&7\\
$t_v \geq$&5&7&8&10&11&13&14\\
$r_v \leq$&25&22&21&19&17&15&14
\end{tabular}
\end{center}
\caption{Privacy numbers and reconstruction numbers of the ramp secret
  sharing scheme described in Example~\ref{ex9}.}
\label{tabpam}
\end{table}

\ \demo
\end{example}
\section{Relatively small codimension}\label{secsmall}

In the former section we demonstrated how to construct good pairs of
nested codes having relatively 
 large codimension. We now show how to obtain good pairs of
nested codes with relatively small codimension. To explain the idea
behind our method, we start with an example which leads to a formal
statement in Theorem~\ref{thesmalltwo} below.

\begin{example}\label{extihi}
This is a continuation of the series of examples where we consider
codes over ${\mathbb{F}}_7^\ast \times {\mathbb{F}}_7^\ast$, and where
we use the degree lexicographic ordering $\prec_{\deg}$. Consider
\begin{eqnarray}
L_1=\{ X^\alpha Y^\beta \in \Delta(6,6) \mid
X^\alpha Y^\beta \preceq_{\deg} XY^3\}, \label{eql1}\\
L_2=\{  X^\alpha Y^\beta \in \Delta(6,6)\mid 
X^\alpha Y^\beta \prec_{\deg} X^3Y \}.\label{eql2}
\end{eqnarray}
The situation is described in Figure~\ref{figflot}

\begin{figure}[h]
$$
\begin{array}{cccccc}
N_{21}&N_{26}&N_{30}&N_{33}&N_{35}&N_{36}\\
N_{15}&N_{20}&N_{25}&N_{29}&N_{32}&N_{34}\\
{\text{\rund{$N_{10}$}}}&{\text{\rectangled{$N_{14}$}}}&N_{19}&N_{24}&N_{28}&N_{31}\\
{\text{\rund{$N_6$}}}&{\text{\rund{$N_9$}}}&{\text{\rectangled{$N_{13}$}}}&N_{18}&N_{23}&N_{27}\\
{\text{\rund{$N_3$}}}&{\text{\rund{$N_5$}}}&{\text{\rund{$N_8$}}}&{\text{\rectangled{$N_{12}$}}}&N_{17}&N_{22}\\
{\text{\rund{$N_1$}}}&{\text{\rund{$N_2$}}}&{\text{\rund{$N_4$}}}&{\text{\rund{$N_7$}}}&{\text{\rund{$N_{11}$}}}&N_{15}
\end{array}
$$
\caption{The situation in Example~\ref{extihi}: The circled monomials
  correspond to $L_2$. The circled and the boxed monomials correspond
  to $L_1$.}
\label{figflot}
\end{figure}
The codimension is $3$ and the values of $u$ and $u^\perp$ in
Theorem~\ref{theour} become
$u=\min \{i \mid N_i \in L_1 \backslash L_2\}=12$ (corresponding to
$N_u=X^3Y$), and $u^\perp=\max \{i \mid N_i \in L_1\}=14$ (corresponding to $N_{u^\perp}=XY^3$). By inspection we see that, due
to the particular choice of $L_1$ and $L_2$, in (\ref{bnd}) and
(\ref{bndperp}) of Theorem~\ref{theour} we need only to consider monomials in $L_1
\backslash L_2$. That is, we obtain
\begin{eqnarray}M_v(C(L_1),C(L_2))&\geq& \min \{ D(K) \mid K \subseteq
\{X^3Y,X^2Y^2,XY^3\}, \# K=v\}, \nonumber \\
M_v(C(L_2)^\perp,C(L_1)^\perp)&\geq& \min \{ D^\perp(K) \mid K \subseteq
\{X^3Y,X^2Y^2,XY^3\}, \# K=v\}, \nonumber
\end{eqnarray}
for $v=1,2,3$. From this we easily obtain the parameters in
Table~\ref{tabthree}.
\begin{table}
\begin{center}
\begin{tabular}{c|rrr}
$v$&1&2&3\\
$M_v(C(L_1),C(L_2))\geq$&15&19&22\\
$M_v(C(L_2)^\perp,C(L_1)^\perp)\geq$&8&11&13
\end{tabular}
\end{center}
\caption{Estimated relative generalized Hamming weights of the first
  code pair in Example~\ref{extihi}}
\label{tabthree}
\end{table}

In a similar way we have in (\ref{bnd}) and
(\ref{bndperp}) only  monomials from $L_1 \backslash
  L_2$, if in (\ref{eql1}) and (\ref{eql2})  we replace
  $(XY^3,X^3Y)$ with $(XY,XY)$, $(XY^2,X^2Y)$, $(X^2Y^2,X^2Y^2)$,
  $(XY^4,X^4Y)$, $(X^2Y^3,X^3Y^2)$, $(X^2Y^4,X^4Y^2)$,
  $(X^3Y^3,X^3Y^3)$, $(X^3Y^4,X^4Y^3)$, or $(X^4Y^4,X^4Y^4)$. However,
  due to symmetry, we only need to consider the first five cases (in addition
  to the case that we have already considered). For instance from
  $(X^4Y^4,X^4Y^4)$ we derive the same estimates for
  $M_v(C(L_1),C(L_2))$ 
and $M_v(C(L_2)^\perp,C(L_1)^\perp)$,
  respectively, as we would derive from $(XY,XY)$ for
  $M_v(C(L_2)^\perp,C(L_1)^\perp)$ and $M_v(C(L_1),C(L_2))$,
  respectively (the order of parameters being reversed). 
In
  Table~\ref{tabtabtab} we list our estimates of relative
  minimum distances and these are compared to the estimates of minimum distances to
  demonstrate the advantage of the proposed code construction. 
\begin{table}
\begin{center}
\begin{tabular}{c|rrrrrr}
$\ell=\#L_1-\#L_2$ &1&2&3&1&4&2\\
$M_1(C(L_1) , C(L_2))\geq$&25&20&15&16&10&12\\ 
$d(C(L_1))\geq$&24&18&12&12&6&6\\
$M_1(C(L_2)^\perp,C(L_1))^\perp\geq$&4&6&8&9&10&12\\
$d(C(L_2)^\perp)\geq$& 3&4&5&5&6&6
\end{tabular}
\end{center}
\caption{The first weights in Example~\ref{extihi}.}
\label{tabtabtab}
\end{table}
In this example we did  in (\ref{eql1}) and (\ref{eql2}), not consider replacing $XY^3$ and $X^3Y$, respectively,
with $Y^a$ and $X^a$, respectively. The
arguments of the example surely would apply also in this case, however, the
corresponding nested codes are of Reed-Muller type, and such codes do
not have impressive parameters. \demo
\end{example}

The method of Example~\ref{extihi} can be applied to any point set
$S_1\times S_2$ with $s=\#S_1=\#S_2$. The idea is to consider the
intersection of a line with slope $-1$ and the lattice
$\Delta(s,s)$. From either direction, both values $D(X^iY^j)$ and
$D^\perp (X^iY^j)$ strictly increase, while moving on the line segment
toward its middle. Hence, choosing $L_2 \subset
L_1$ in such a way that $L_1 \backslash L_2$ is equal to a center part
of a line segment of slope $-1$ produces good relative minimum distances.

\begin{theorem}\label{thesmalltwo}
Consider $S_1,S_2 \subseteq {\mathbb{F}}_q$ with $s=\# S_1 =\#
S_2$. Let $I=\langle F_1(X),F_2(Y)\rangle \subset {\mathbb{F}}_q[X,Y]$,
where $F_1$ and $F_2$ are as in~(\ref{eqfi}). Consider $X^iY^j \in \Delta(s,s)$ with
$i \leq j$ and let
\begin{eqnarray}
L_1&=&\{ N \in \Delta(s,s)\mid N \preceq_{\deg}X^iY^j\}, \label{eqgod1}\\
L_2&=&\{ N \in \Delta(s,s)\mid N \prec_{\deg}X^jY^i\}.\label{eqgod2}
\end{eqnarray}
The codes $C(L_1)$ and $C(L_2)$ are of length $n=s^2$ and their
codimension equals $\ell= j-i+1$. The relative minimum distances satisfy
\begin{eqnarray}
M_1(C(L_1),C(L_2))& =& (s-i)(s-j) \label{eqm1}, \\
M_1(C(L_2)^\perp , C(L_1)^\perp  &\geq &(i+1)(j+1), \label{eqm2}
\end{eqnarray}
and for $v=2, \ldots , \ell$
\begin{eqnarray}
M_v(C(L_1),C(L_2))& =& (s-i)(s-j)+\sum_{t=2}^v \big(
                       (s-i)-(t-1)\big)\label{eqcarnando1} \\
&=&(s-i)(s-j+v-1)-\frac{v(v-1)}{2}, \nonumber \\
M_v(C(L_2)^\perp,C(L_1)^\perp) &\geq& (i+1)(j+1)+\sum_{t=2}^v
                                      \big((j+1)-(t-1)\big)\label{eqcarnando2} \\
&=&(i+v)(j+1)-\frac{v(v-1)}{2}.\nonumber
\end{eqnarray}
\end{theorem}
\begin{proof}
We first establish the bounds 
\begin{eqnarray}
M_v(C(L_1),C(L_2))& \geq& \min \big\{ D(K) \mid K \subseteq \{ X^jY^i,
  X^{j-1}Y^{i+1},\nonumber \\
&& \ldots , X^iY^j\}, \# K =v\big \}, \label{eqgenbound1} \\
M_v(C(L_2)^\perp,C(L_1)^\perp) &\geq& \min \big \{ D^\perp(K) \mid K \subseteq \{ X^jY^i,
  X^{j-1}Y^{i+1},\nonumber \\
&& \ldots , X^iY^j\}, \# K =v\big\}. \label{eqgenbound2} 
\end{eqnarray}
We do this by applying
Theorem~\ref{theour} with $\prec_{\deg}$ as the chosen monomial
ordering. In particular $\Delta(s,s)=\{ N_1, \ldots , N_{s^2}\}$ where
the enumeration of the $N_i$s is with respect to $\prec_{\deg}$. Consider $u=\min \{i \mid N_i \in L_1 \backslash L_2\}$ and $u^\perp =
\max \{i \mid N_i \in L_1 \}$. The $u$th element of $\Delta(s,s)$ now
is $N_u=X^jY^i$ and the $u^\perp$th element is
$N_{u^\perp}=X^iY^j$. By~(\ref{bnd}) and (\ref{bndperp}) the right-hand sides of
(\ref{eqgenbound1}) and (\ref{eqgenbound2}) therefore serve as lower
bounds on their respective left-hand sides.\\
We next prove~(\ref{eqm1}) and (\ref{eqm2}).
 From~(\ref{eqgenbound1}) we have
\begin{eqnarray}
&&M_1(C(L_1),C(L_2)) \nonumber \\
&\geq &\min \{ D(X^jY^i),
D(X^{j-1}Y^{i+1}),\ldots , D(X^iY^j)\}\nonumber\\
&=&\min\{
(s-(i+v))(s-(j-v)) \mid v=0, \ldots , \ell-1=j-i\}\nonumber \\
&=&(s-i)(s-j)\nonumber
\end{eqnarray}
(equality of the convex function is attained for $v=0$ and
$v=\ell-1$). This proves that the right-hand side of (\ref{eqm1}) is larger
than or equal to the left-hand side. In a similar fashion we establish the
inequality~(\ref{eqm2}). To establish equality in~(\ref{eqm1}) we only
need to find a codeword in $C(L_1)\backslash C(L_2)$ with $(s-i)(s-j)$
non-zeros. To this end, write $S_1=\{\alpha_1, \ldots ,
\alpha_s\}$ and $S_2=\{ \beta_1, \ldots , \beta_s\}$ and consider 
$$F(X,Y)=\prod_{w=1}^i (X-\alpha_w)\prod_{r=1}^j (Y-\beta_r)$$ 
which has exactly the prescribed number of non-zeros in $S_1\times
S_2$. The codeword $\vec{c}={\mbox{\textnormal{ev}}}(F(X,Y)+I)$ clearly belongs to
$C(L_1) \backslash C(L_2)$ and therefore we have established equality
in~(\ref{eqm1}).\\
In a similar way we can show that equality actually holds
in~(\ref{eqgenbound1}). This is done by considering for each possible
set $K\subseteq \{X^jY^i, 
X^{j-1}Y^{i+1}, \ldots , X^iY^j\}$ the corresponding set of $\#K$
polynomials as above. The number of elements in $S_1\times S_2$, that
are non-zeros of at least one of these polynomials, equals
$D(K)$.\\
It remains to show that the right-hand side of~(\ref{eqgenbound1}) equals
the right-hand side of~(\ref{eqcarnando1}), and that the right-hand side
of~(\ref{eqgenbound2}) equals the right-hand side
of~(\ref{eqcarnando2}). For convenience, we will explain the case
in~(\ref{eqcarnando1}) for $v=2$. The cases $v \geq 3$ can be proved
with the same reasoning by using the principle of inclusion and
exclusion. A similar and straightforward reasoning, where one should
replace $(s-i)(s-j)$ with $(i+1)(j+1)$ gives the corresponding formula
for the case in~(\ref{eqcarnando2}).\\
Indeed, we are considering sets $K=\{X^{j-x}Y^{i+x},X^{j-y}Y^{i+y}\}$,
where one can assume $0 \leq x < y\leq j-i=\ell-1$. Then $D(K)$ equals
the cardinality of the set of monomials which are divisible by either
$X^{j-x}Y^{i+x}$ 
or 
$X^{j-y}Y^{i+y}$. So, we have to compute the sum
of the cardinalities of the set of monomials divisible by
$X^{j-x}Y^{i+x}$ plus that of the set of monomials divisible by
$X^{j-y}Y^{i+y}$ minus that of the set of monomials divisible by both
of them. Therefore,
\begin{eqnarray}
D(K)&=&\big((s-j)+x\big)\big((s-i)-x\big)+\big((s-j)+y\big)\big((s-i)-y\big)\nonumber
   \\
&&-\big((s-j)+x\big)\big((s-i)-y\big).\nonumber
\end{eqnarray}
We are looking for the minimum of that function within the above
mentioned region. Derivatives show that the minimum will appear on the
boundary. As a consequence we find it for the values $x=0$ and $y=1$
which give the value $(s-i)(s-j)+(s-i-1)$ in the statement. This
concludes the proof.
\end{proof}

\begin{proposition}\label{corsmalll}
Let the notation be as in Theorem~\ref{thesmalltwo} and consider
$\sigma \in \{0, \ldots , s-1\}$. Then for any positive integer $\ell
\leq \sigma +1$ such that $\ell$ is even if and only if $\sigma$ is
odd we obtain nested code pairs $C(L_2) \subset C(L_1)$ of
codimension $\ell$ with
\begin{eqnarray}
M_1(C(L_1),C(L_2))&=&\bigg(s-\frac{\sigma -\ell+1}{2}\bigg)\bigg(s-\frac{\sigma
                      +\ell-1}{2}\bigg) \label{eqqqsnabel1}\\
M_1(C(L_2)^\perp,C(L_1)^\perp)&\geq &\bigg(\frac{\sigma -\ell +3}{2}\bigg)\bigg(\frac{\sigma
                      +\ell+1}{2}\bigg) \label{eqqqsnabel2}
\end{eqnarray}
and if $\ell \leq \sigma-1$ then
\begin{eqnarray}
d(C(L_1))&=&s(s-\sigma)<M_1(C(L_1),C(L_2)).\label{eq2snabel2.5}
\end{eqnarray}
Similarly we obtain code pairs of codimension $\ell$ with 
\begin{eqnarray}
M_1(C(L_1),C(L_2))&=&\bigg(\frac{\sigma -\ell +3}{2}\bigg)\bigg(\frac{\sigma
                      +\ell+1}{2}\bigg)  \label{eqqqsnabel4} \\
                     M_1(C(L_2)^\perp,C(L_1)^\perp)&\geq &\bigg(s-\frac{\sigma -\ell+1}{2}\bigg)\bigg(s-\frac{\sigma +\ell-1}{2}\bigg) \label{eqqqsnabel3}
\end{eqnarray}
and if $\ell \leq \sigma-1$ then
\begin{eqnarray}
d(C(L_1))=\sigma+1 < M_1(C(L_1),C(L_2)).\nonumber 
\end{eqnarray}
\end{proposition}
\begin{proof}
We only prove (\ref{eqqqsnabel1}), (\ref{eqqqsnabel2}) and (\ref{eq2snabel2.5}). The other
results follow by symmetry. Choose $0 \leq i \leq j < s$ with
$i+j=\sigma < s$ and $\ell =j-i+1$. Then $i=\frac{\sigma - \ell
  +1}{2}$, $j=\frac{\sigma +\ell-1}{2}$ and the first two results
follow from Theorem~\ref{thesmalltwo}. Finally, $d(C(L_1))=M_1(C(L_1),\{\vec{0}\})=D(X^\sigma)=s(s-\sigma)$.
\end{proof}
\begin{corollary}\label{corrrr}
Consider integers $1 < s  \leq q$, where $q$ is a prime power and let
$\sigma \in \{0, \ldots , s-1\}$. Then for any $\ell \leq \sigma +1$
such that $\ell$ is even if and only if $\sigma$ is odd we obtain from
Proposition~\ref{corsmalll} ramp secret sharing schemes over
${\mathbb{F}}_q$ with $n=s^2$ participants, shares in
${\mathbb{F}}_q^\ell$ and either
\begin{eqnarray}
t=t_1&\geq & \bigg( \frac{\sigma -\ell+3}{2}\bigg)\bigg( \frac{\sigma
             +\ell+1}{2}\bigg)-1\nonumber \\ 
r=r_\ell&=&s^2-\bigg( s-\frac{\sigma -\ell+1}{2}\bigg)\bigg(
            s-\frac{\sigma +\ell-1}{2}\bigg)+1 \nonumber 
\end{eqnarray}
or
\begin{eqnarray}
t=t_1&\geq & \bigg( s-\frac{\sigma -\ell+1}{2}\bigg)\bigg(
            s-\frac{\sigma +\ell-1}{2}\bigg)-1\nonumber \\ 
r=r_\ell&=&s^2-\bigg( \frac{\sigma -\ell+3}{2}\bigg)\bigg( \frac{\sigma
             +\ell+1}{2}\bigg)+1. \nonumber  
\end{eqnarray}
Similarly we obtain asymmetric quantum codes with parameters
$$[[n=s^2, \ell, d_z/d_x]]_q$$
where $d_z$ equals the right hand side of~(\ref{eqqqsnabel1}) and
$d_x$ is greater than or equal to the right hand side
of~(\ref{eqqqsnabel2}) (and similarly with (\ref{eqqqsnabel4}) and
(\ref{eqqqsnabel3})). If $\ell \leq \sigma-1$ then the asymmetric
quantum codes are impure.
\end{corollary}

\begin{remark}
As mentioned in the introduction it is often desirable to have
asymmetric quantum codes with $d_z$ much larger than $d_x$. One such family is
obtained from Corollary~\ref{corrrr} with parameters $[[n=s^2, \ell, d_z \geq
s(s-\ell+1) / d_x =\ell ]]_q$ for any $\ell\in \{1, \ldots ,
s-1\}$. For $q=7,8,9$ and $\ell=1,2,3,4,5$ these codes strictly exceed
the Gilbert-Varshamov bound (Theorem~\ref{theMGV}). Similarly for
$q=5,11,13,16,17,19,23,25$ and $\ell=1,2,3,4$.
\end{remark}

\begin{example}\label{exandregras}
In this example we consider asymmetric quantum codes as in
Corollary~\ref{corrrr} with $d_z=\delta$ being equal to the right hand
side of (\ref{eqqqsnabel1}) and $d_x$ being greater than or equal to
$\delta^\perp$ which we define as the right hand side
of~(\ref{eqqqsnabel2}). We treat the cases
$n=q^2$, where $q=3,4,5,7,8,9$. 
In the
literature e.g.~\cite{aq1,ezerman2013css,aqlaguardia} one can find extensive tables of
quantum stabilizer code parameters derived by applying
Theorem~\ref{thecss}, however, they only use the bound $d_z\geq d(C_1)$ and
$d_x\geq d(C_2^\perp)$, where $d(C_1)$ and $d(C_2^\perp)$ are the
minimum distances of concrete code pairs with $C_2 \subset C_1$. The
present example illustrates the huge advantage of using instead the
relative minimum distances (which is what is behind the bounds in Corollary~\ref{corrrr}). This is done by investigating for each $\ell$ and
$\delta$ what is the highest value $g(\ell, \delta)$ such
that the tables of best known linear codes in~\cite{grassl} guarantee the existence of linear
code pairs $A, B^\perp$ satisfying $\dim A-\dim B=\ell$, $d(A) \geq
\delta$, and $d(B^\perp)\geq g(\ell,\delta)$ (this is in the spirit
of~\cite[Theorem 2]{rev3_5}). Observe, that we make no
assumption whatsoever that $B \subset A$. Actually, such inclusion is
very unlikely to hold when one chooses two codes $A$ and $B^\perp$ which are
optimal with respect to the tables of best known linear codes in~\cite{grassl}. In
Table~\ref{tabgrasgroent} we list values of $(\ell, \delta,
\delta^\perp,g(\ell, \delta))$. The many cases
where $\delta^\perp$ is close to $g(\ell,\delta)$ illustrate the huge
advantage of using the construction in Theorem~\ref{thesmalltwo} and
taking into account the relative minimum distances. Note that, there
are even two
cases where $\delta^\perp$ exceeds the corresponding $g(\ell,\delta)$,
namely for $q=7$ and $(\ell, \delta, \delta^\perp, g(\ell,\delta))$
equal to $(3, 15, 15, 14)$ or $(2, 30, 6, 5)$.  All displayed code parameters coming from Theorem~\ref{thesmalltwo} strictly exceed the
Gilbert-Varshamov bound (Theorem~\ref{theMGV}).

\begin{table}
\begin{center}
\begin{tabular}{c|l}
$q$&$(\ell,\delta,\delta^\perp,g(\ell,\delta))$\\
\hline
3&(1,4,4,4)\\
\ \\
4&(2,6,6,6) (1,9,4,4)\\
\ \\
5&(3,8,9,9) (1,9,9,9) (2,12,6,6) (1,16,4,4)\\
\ \\
7&(5,12,12,18) {\bf{(3,15,15,14)}} (1,16,16,17) (4,18,10,13)\\
&(2,20,12,12) (3,24,8,9) (1,25,9,10) {\bf{(2,30,6,5)}}\\
&(1,36,4,4)\\
\ \\
8&(5,21,12,19) (3,24,15,16) (1,25,16,16) (4,18,18,21) \\
&(2,20,20,22) (4,28,10,13) (2,30,12,13) (3,35,8,8)\\
& (1,36,9,10) (2,42,6,6) (1,49,4,4)\\
\ \\
9&(3,24,24,26) (1,25,25,26) 
(5,21,21,27) (4,28,18,21)\\
& (2,30,20,22) (5,32,12,18)
(3,35,15,16) (1,36,16,16)\\
& (4,40,10,13) (2,42,12,14) (3,48,8,8) (1,49,9,9)\\
& (2,56,6,6) (1,64,4,4)
\end{tabular}
\end{center}
\caption{Corresponding values of $(\ell, \delta, \delta^\perp,
  g(\ell,\delta))$ from Example~\ref{exandregras}. The many cases with
$\delta^\perp$ close to $g(\ell,\delta)$ (and even {\bf{two cases}} with
$\delta^\perp >g(\ell,\delta)$) demonstrate the advantage of the
construction in Theorem~\ref{thesmalltwo} and of using relative
minimum distances.}
\label{tabgrasgroent}
\end{table}
\ \demo
\end{example}

It is not straightforward to generalize 
Theorem~\ref{thesmalltwo} to  
point ensembles $S_1
\times S_2$ with $\# S_1 $ not necessarily equal to $ \# S_2$. The problem lies
in the choice of monomial ordering (and the corresponding definition of
$L_1$ and $L_2$). More concretely, for $\#S_1 \neq \# S_2$ there simply is no monomial ordering
which simultaneously optimizes the relative minimum distance of the
corresponding primary
codes and the relative minimum distance of the corresponding dual
codes. However, our method can still be applied as the following
example illustrates.
\begin{example}\label{ex101}
In this example we consider $S_1, S_2 \subseteq {\mathbb{F}}_q$ with
$s_1=\# S_1=8$, $s_2=\# S_2=5$ and consider codes defined from $S_1
\times S_2$. Here $q$ is any prime power greater than or equal to $8$. In Figure~\ref{figfjong} we depict on the left
$D(\Delta(s_1,s_2))$ and on the right
$D^\perp(\Delta(s_1,s_2))$. 
\begin{figure}
\begin{center}
$\begin{array}{rrrrrrrr}
8&7&6&5&4&3&2&1 \\
16&14&12&10&8&6&4&2\\
24&21&18&15&12&9&6&3\\
32&28&24&20&16&12&8&4\\
40&35&30&25&20&15&10&5
\end{array}
\ \ \ 
\begin{array}{rrrrrrrr}
5&10&15&20&25&30&35&40\\
4&8&12&16&20&24&28&32\\
3&6&9&12&15&18&21&24\\
2&4&6&8&10&12&14&16\\
1&2&3&4&5&6&7&8
\end{array}$
\end{center}
\caption{Example~\ref{ex101}: $D(N)$ to the left, and $D^\perp(N)$ to the right} 
\label{figfjong}
\end{figure}

Concentrating first on the lower left
corner of $\Delta(s_1,s_2)$ we see that on a line segment of slope
$-1$ the values $D^\perp(X^iY^j)$ indeed still increase when moving
from either direction toward the middle of the segment. This,
however, does not at all hold for $D(X^iY^j)$. For instance if we
choose 
$$L_1=\{N \in \Delta(s_1,s_2) \mid N \preceq_{\deg} XY^2\},$$
$$L_2=\{N \in \Delta(s_1,s_2) \mid N \prec_{\deg} X^2Y\}, $$
then in (\ref{bnd}) and (\ref{bndperp}) of Theorem~\ref{theour} we only need
to consider monomials in $L_1 \backslash L_2=\{X^2Y,XY^2\}$. Hence,
we obtain the estimates $M_1(C(L_1),C(L_2))\geq \min \{21,24\}=21$ 
and
$M_1(C(L_2)^\perp,C(L_1)^\perp)\geq 6$. But this seems somewhat not a
perfect choice of $L_2 \subset L_1$ as now $\min \{ D(M) \mid M\in
L_1\}=\min \{
D(M) \mid M \in L_1\backslash L_2\}$. Hence, we optimized the relative minimum distance of
the dual codes, but did not obtain any improvement for the primary
codes. Turning to the right upper corner of $\Delta(s_1,s_2)$ the
situation is similar, however, with the role of the primary and dual codes
interchanged. We finally consider the remaining middle part of
$\Delta(s_1,s_2)$. We first choose as monomial ordering the weighted
degree lexicographic ordering $\prec_w$ defined by the rule that  $X^{i_1}Y^{i_2} \prec_w X^{j_1}Y^{j_2}$ if
either $i_1+2i_2 < j_1+2j_2$, or $i_1+2i_2=j_1+2j_2$ with  $i_2 < j_2$. Then define
$$L_1=\{ N \in \Delta(s_1,s_2) \mid N \preceq_wX^2Y^2\},$$
$$L_2=\{ N \in \Delta(s_1,s_2) \mid N \prec_wX^4Y\}.$$
Again this renders the nice property that in (\ref{bnd}) and (\ref{bndperp}) we only
need to consider the monomials of $L_1\backslash L_2$, which in this
case becomes $\{X^4Y,X^2Y^2\}$. We obtain 
$$M_1(C(L_1),C(L_2)) \geq
\min \{16, 18\}=16,$$
$$M_1(C(L_2)^\perp, C(L_1)^\perp)\geq \min \{9,10\}=9.$$ Choosing on
the other hand the degree lexicographic ordering and defining
$$L_1=\{N \in \Delta(s_1,s_2) \mid N \preceq_{\deg} X^2Y^2\},$$
$$L_2=\{N \in \Delta(s_1,s_2) \mid N \prec_{\deg} X^3Y\},$$
we only need to consider monomials in $L_1\backslash
L_2=\{X^3Y,X^2Y^2\}$, from which we obtain $M_1(C(L_1),C(L_2))\geq \min
\{18,20\}=18$ and $M_1(C(L_2)^\perp, C(L_1)^\perp ) \geq \min \{ 8,9
\}=8$. As a consequence, there seems to be no general rule for which (weighted) degree
lexicographic ordering to choose.\demo
\end{example}
We now return to the case of the point set being a Cartesian product
of subsets $S_i\subseteq {\mathbb{F}}_q$ of the same
size. Theorem~\ref{thesmalltwo} treated the two-dimensional case, the
theorem below treats higher dimensions.
\begin{theorem}\label{thehigherdim}
Consider $S_1, \ldots , S_m \subseteq {\mathbb{F}}_q$ with $s=\# S_1=\cdots
= \#S_m$ and let\\
 $\langle F_1(X_1), \ldots , F_m(X_m)\rangle \subset
{\mathbb{F}}_q[X_1, \ldots , X_m]$, where $F_1, \ldots  F_m$ are as in
(\ref{eqfi}). Consider $X_1^{i_1}X_2^{i_2}\cdots X_m^{i_m} \in \Delta(s,
\ldots , s)$ with $i_1 \leq i_2$ and let
$$L_1=\{ N \in \Delta(s,\ldots , s)\mid N
\preceq_{\deg}X_1^{i_1}X_2^{i_2} \cdots  X_m^{i_m}\},$$
$$L_2=\{ N \in \Delta(s,\ldots , s) \mid N
\prec_{\deg}X_1^{i_2}X_2^{i_1}X_3^{i_3} \cdots , X_m^{i_m}\}.$$
The codes $C(L_1)$ and $C(L_2)$ are of length $n=s^m$ and their
codimension equals $\ell= i_2-i_1+1$. The relative minimum distances satisfy
\begin{eqnarray}
M_1(C(L_1),C(L_2)) &=& (s-i_1)\cdots (s-i_m) \nonumber \\
M_1(C(L_2)^\perp , C(L_1)^\perp  &\geq& (i_1+1)\cdots (i_m+1),
                                        \nonumber 
\end{eqnarray}
and for $v=2, \ldots , \ell$
\begin{eqnarray}
M_v(C(L_1),C(L_2)) &=& 
\min \big\{ D(K) \mid 
K \subseteq \{
  X_1^{i_2}X_2^{i_1}X_3^{i_3}\cdots X_m^{i_m},\nonumber \\
&&  \ldots , X_1^{i_1}X_2^{i_2}\cdots X_m^{i_m} \}, \# K =v\big\},
   \nonumber \\
M_v(C(L_2)^\perp,C(L_1)^\perp) &\geq& 
 \min \big\{ D^\perp(K) \mid 
K \subseteq \{
  X_1^{i_2}X_2^{i_1}X_3^{i_3}\cdots X_m^{i_m},\nonumber \\
&&  \ldots , X_1^{i_1}X_2^{i_2}\cdots X_m^{i_m} \}, \# K =v\big\}.\nonumber
\end{eqnarray}
\end{theorem}
\begin{proof}
The proof is similar to that of Theorem~\ref{thesmalltwo}. The details
are left for the reader.
\end{proof}

We next return to the two-dimensional case, comparing in an example
asymmetric quantum codes from Corollary~\ref{corrrr} 
with La Guardia's Construction II of asymmetric
quantum generalized Reed-Solomon codes
(Theorem~\ref{thelaguardia}). Recall from Remark~\ref{remsammenhaeng}
that parameters of asymmetric quantum codes based on the CSS
construction can be directly translated into parameters of ramp
secret sharing schemes.
\begin{example}\label{exq2}
We first consider the case of asymmetric quantum codes with $q=7$ and of length
$49$. The parameters produced by Theorem~\ref{thelaguardia} for the
code length $49$ all satisfy that $d_z, d_x \leq 6$. To produce higher
values of $d_z$ (as usual we shall assume $d_z \geq d_x$), we can instead apply
Theorem~\ref{thelaguardia} to codes of length $48$ and thereby
derive information on codes of length $49$. We then compare these
values with 
what is produced from Theorem~\ref{thesmalltwo} in combination with Theorem~\ref{thecss} for codes of length
$49$. As is seen in Table~\ref{tablaguardia3}, most often
Theorem~\ref{thesmalltwo} produces the best results. All displayed
codes coming from Theorem~\ref{thesmalltwo} strictly exceed the
Gilbert-Varshamov bound (Theorem~\ref{theMGV}).\\
\begin{table}
\begin{center}
\begin{tabular}{cc}
Theorem \ref{thelaguardia} (\cite[Theorem 7.1]{aqlaguardia})&Theorem \ref{thesmalltwo}\\
\hline
-----&$[[49,1,31/4]]_7$\\
-----&$[[49,2,30/6]]_7$\\
-----&$[[49,1,25/9]]_7$\\
$[[49,2,23/2]]_7$&$[[49,3,24/8]]_7$\\
$[[49,2,20/5]]_7$&$[[49,2,20/12]]_7$\\
$[[49,2,18/7]]_7$&$[[49,4,18/10]]_7$\\
$[[49,2,16/9]]_7$&$[[49,1,16/16]]_7$\\
$[[49,2,15/10]]_7$&$[[49,3,15/15]]_7$\\
$[[49,6,12/11]]_7$&$[[49,5,12/12]]_7$
\end{tabular}
\end{center}
\caption{Comparison of code parameters in Example~\ref{exq2}. To the
  right all possible parameters derived using
  Theorem~\ref{thesmalltwo}. To the left a selection of parameters
  derived by applying Theorem~\ref{thelaguardia} to codes of length
  $48$. An empty entry means that there are no comparable parameter.}
\label{tablaguardia3}
\end{table}
We next consider asymmetric quantum codes with $q=8$ and of length
$64$. In 
Table~\ref{tablaguardia4} we compare representative examples of what
can be derived from Theorem~\ref{thelaguardia} with what can be obtained from
Theorem~\ref{thesmalltwo} in combination with
Theorem~\ref{thecss}. Again the advantage of our method is distinct in
most cases, however, with a clear exception when $d_z=14$. For $d_z > 31$,
Theorem~\ref{thelaguardia} does not produce any information, which in
Table~\ref{tablaguardia4} is marked with ``-----''. All displayed codes coming
from Theorem~\ref{thesmalltwo} strictly exceed the
Gilbert-Varshamov bound (Theorem~\ref{theMGV}).\\

\begin{table}
\begin{center}
\begin{tabular}{cc}
Theorem \ref{thelaguardia} (\cite[Theorem 7.1]{aqlaguardia})&Theorem~\ref{thesmalltwo}\\
\hline
----- \ \ \ \ \ \ \ &$[[64,1,49/4]]_8$\\
----- \ \ \ \ \ \ \ &$[[64,2,42/6]]_8$\\
----- \ \ \ \ \ \ \ &$[[64,1,36/9]]_8$\\
----- \ \ \ \ \ \ \ &$[[64,3,35/8]]_8$\\
$[[64,2,30/3]]_8$&$[[64,2,30/12]]_8$\\
$[[64,4,28/4]]_8$&$[[64,4,28/10]]_8$\\
$[[64,2,25/8]]_8$&$[[64,1,25/16]]_8$\\
$[[64,2,24/9]]_8$&$[[64,3,24/15]]_8$\\
$[[64,2,21/12]]_8$&$[[64,5,21/12]]_8$\\
$[[64,2,20/13]]_8$&$[[64,2,20/20]]_8$\\
$[[64,4,18/14]]_8$&$[[64,4,18/18]]_8$
\end{tabular}
\end{center}
\caption{Comparison of code parameters in Example~\ref{exq2}. To the
  right all possible parameters from Theorem~\ref{thesmalltwo}. To the
  left a selection of parameters resulting from Theorem~\ref{thelaguardia}.}
\label{tablaguardia4}
\end{table}
\ \demo
\end{example}\label{remsidste}
We conclude this section with discussing higher weights and their
use in secret sharing. 
\begin{remark}\label{remarkenN}
Inspecting~(\ref{eqcarnando1}) and
(\ref{eqcarnando2}) it is clear that the $(v+1)$th relative weights
are typically much larger than the $v$th relative weights, for $v \in \{1,
\ldots , \ell-1\}$. In particular the second relative weights are often much larger
than the first relative weights. The consequence for the related ramp
secret sharing schemes is that the security is much better than what
is reflected only by the parameter $t=t_1$, in that $t_{v+1}$ is much
larger than $t_v$ for $v \in \{1, \ldots , \ell-1\}$. Hence, if a small
amount of information leakage can be accepted, then one can tolerate
many more leaked symbols. In the other direction, reconstruction
corresponds to solving a system of linear equations. Hence, the fact that $r_v$ is
much smaller than $r_{v+1}$ for $v \in \{1, \ldots , \ell-1\}$, and in
particular that $r_{\ell-1}$ is much smaller than $r=r_\ell$, means
that if one is willing to guess some of the indeterminates of the
system then one needs much fewer shares to reconstruct the secret. 
\end{remark}
We illustrate Remark~\ref{remarkenN} with an example.
\begin{example}\label{exnoyes}
This is a continuation of Example~\ref{extihi}. Applying
Theorem~\ref{thesmalltwo} we obtain ramp secret sharing schemes with
$n=36$ participants, with secrets in ${\mathbb{F}}_q^\ell$ where
$\ell$ and the privacy and reconstruction numbers are as in
Table~\ref{tabnuogsaa}.
\begin{table}
\begin{center}
\begin{tabular}{ccccccccc}
$\ell$&$t_1$&$t_2$&$t_3$&$t_4$&$r_1$&$r_2$&$r_3$&$r_4$\\
\hline 
1&24&-&-&-&33&-&-&-\\
2&19&23&-&-&29&31&-&-\\
3&14&18&21&-&24&26&29&-\\
1&15&-&-&-&28&-&-&-\\
4&9&13&16&18&18&20&23&27\\
1&8&-&-&-&21&-&-&-\\
3&7&10&12&-&15&18&22&-\\
2&5&7&-&-&13&17&-&-\\
1&3&-&-&-&12&-&-&-
\end{tabular}
\end{center}
\caption{Ramp secret sharing schemes from Example~\ref{exnoyes}}
\label{tabnuogsaa}
\end{table}

\ \demo
\end{example}

\section{Concluding remarks}\label{secon}
In a series of works~\cite{galindo2015quantum,galindo2015quantumwithout,galindo2015stabilizer,galindo2015new,allefire} the 
authors of the present paper investigated linear codes over ${\mathbb{F}}_q$
defined by evaluating multivariate polynomials at Cartesian products
$S_1\times \cdots \times S_m$, where for $i=1,\ldots , m$, $S_i$ is 
the set of roots of
\begin{eqnarray}
X_i^{N_i}-X_i& {\mbox{  \ or \ }}&X_i^{N_i-1}-1. \label{eqhelix}
\end{eqnarray}
Here $N_i >1$ satisfies that $N_i-1$ divides $q-1$. Such codes were
then used in~\cite{galindo2015quantum,galindo2015quantumwithout,galindo2015stabilizer,galindo2015new,allefire} for the construction of symmetric quantum codes. In the
terminology used in these papers a set $J \subseteq \{1, \ldots , m\}$ indicates for which
indices the second case in~(\ref{eqhelix}) occurs -- and the
corresponding codes are called $J$-affine variety codes, and if
$N_i-1=q-1$, $J=\{1, \ldots , m\}$ generalized toric codes~\cite{diego}. One of the
advantages of such codes is that they come with an efficient method
for finding parity check matrices. More precisely, when each row in
the generator matrix is made by evaluating a monomial at the points
of the point set, then~\cite[Proposition 1]{galindo2015stabilizer} provides a description
  of a corresponding parity check matrix. The codes of the present
  paper clearly are $J$-affine variety codes when $S_i$, $i=1, \ldots ,
  m$ is of the form~(\ref{eqhelix}). This is in particular the case
  in all the examples we have given, implying that for these codes we
  can easily establish parity check matrices.

Another
  advantage of $J$-affine variety codes is that they are suited for the
  construction of subfield subcodes. It is an interesting topic of future research
  to investigate if the  method from the present paper
can be successfully combined with such subfield subcode
construction. 

For $q$ an even power of a prime, Theorem~\ref{thecss} is also true if one
replaces the 
Euclidean duality with the Hermitian duality~\cite[Theorem~4.5]{aq5}. It
is also an interesting research problem to investigate if this product can be
successfully combined with the methods of the present paper.\\

\section*{Acknowledgments}
The authors thank Ryutaroh Matsumoto for pleasant
discussions and the anonymous reviewers for insightful comments
and suggestions.  
The authors gratefully acknowledge the support from The Danish Council for Independent Research (Grant
No.\ DFF--4002-00367), the support 
from the Spanish MINECO/FEDER
(Grants No.\ MTM2015-65764-C2-2-P and MTM2015-69138-REDT), and the support from University
Jaume I (Grant No.\ P1-1B2015-02).

\appendix

\section{Proof of Lemma~\ref{lemuseful}}\label{apa}
\begin{proof}
Let $m\geq 2$ be an arbitrary integer. The proof is by induction on
$i=m, \ldots , 1$. For $i=m$ the formula reduces to 
$$\int_0^{s-\frac{\tau}{(s-x_1)\cdots
    (s-x_{m-1})}}dx_m=s-\frac{\tau}{(s-x_1)\cdots (s-x_{m-1})}$$
which is indeed true. Next let $1 \leq i < m$ and assume that the
formula in the lemma holds when $i$ is substituted with $i+1$. We must
show that it also holds for $i$. The left hand side becomes
\begin{eqnarray}
&&\int_0^{s-\frac{\tau}{s^{m-i}(s-x_1)\cdots
  (s-x_{i-1})}}s^{m-i}\nonumber \\
&&-\sum_{t=0}^{m-i-1}
  \frac{1}{t!}\frac{\tau}{(s-x_1) \cdots (s-x_i)}\bigg(\ln
  \bigg(\frac{(s-x_1)\cdots (s-x_i)s^{m-i}}{\tau}\bigg)\bigg)^t \,
  dx_i \nonumber \\
&=&s^{m-i+1}-\frac{\tau }{(s-x_1)\cdots (s-x_{i-1})} \nonumber \\
&&+\sum_{t=0}^{m-i-1} \int_0^{s-\frac{\tau}{s^{m-i}(s-x_1)\cdots
   (s-x_{i-1})}}\frac{1}{t!}\frac{\tau}{(s-x_1)\cdots
   (s-x_{i-1})}\frac{-1}{(s-x_i)}\nonumber \\
&&\cdot \bigg( \ln \bigg( \frac{(s-x_1)\cdots
   (s-x_i)s^{m-i}}{\tau}\bigg)\bigg)^t \, dx_i.\label{eqterms}
\end{eqnarray}
We continue with the last term, which after the substitution, 
$$u=\ln ( (s-x_1)\cdots (s-x_i)s^{m-i}/\tau),$$
becomes
\begin{eqnarray}
&&\sum_{t=0}^{m-i-1} \int_{\ln \bigg(\frac{(s-x_1)\cdots
  (s-x_{i-1})s^{m-i+1}}{\tau}\bigg)}^0
  \frac{1}{t!}\frac{\tau}{(s-x_1)\cdots (s-x_{i-1})}u^t \, du\nonumber
  \\
&=&-\sum_{t=0}^{m-i-1} \frac{1}{(t+1)!}\frac{\tau}{(s-x_1) \cdots
    (s-x_{i-1})}\bigg( \ln \bigg(\frac{(s-x_1) \cdots
    (s-x_{i-1})s^{m-i+1}}{\tau}\bigg) \bigg)^{t+1}.\nonumber
\end{eqnarray}
Shifting the index by $1$ in the last sum and collecting terms
in~(\ref{eqterms}) prove the lemma.
\end{proof}

\end{document}